%% file: energy-calculus.tex
\documentclass[10pt]{article}
\usepackage[margin=1in]{geometry}

\usepackage{amsmath,amssymb,amsthm}
\usepackage{mathtools}

\usepackage{stix2}

\usepackage{mlenergy}

\usepackage{graphicx}
\usepackage{wrapfig}
\usepackage{subfig}

\usepackage{array}
\usepackage{booktabs}
\usepackage{enumitem}
\usepackage[utf8]{inputenc}
\usepackage[T1]{fontenc}

\usepackage{cite, xurl}
\usepackage{hyperref}
\hypersetup{
  colorlinks=true,
  linkcolor=mlelink,
  citecolor=mlelink,
  filecolor=mlelink,
  urlcolor=mlelink
}
\usepackage[nameinlink]{cleveref}
\crefname{section}{Section}{Sections}
\Crefname{section}{Section}{Sections}
\crefname{appendix}{Appendix}{Appendices}
\Crefname{appendix}{Appendix}{Appendices}
\crefname{property}{Property}{Properties}
\Crefname{property}{Property}{Properties}

\newtheorem{definition}{Definition}[section]
\newtheorem{property}{Property}[section]
\newtheorem{proposition}{Proposition}[section]
\newtheorem{corollary}{Corollary}[section]
\newtheorem{axiom}{Axiom}[section]

\newcommand{\seqop}{\otimes}
\newcommand{\parop}{\oplus}
\newcommand{\parml}{\oplus_{\text{multi}}}

\newcommand{\Epred}{E_{\text{predicted}}}
\newcommand{\Emeas}{E_{\text{measured}}}
\newcommand{\Pcap}{P_{\text{cap}}}

\newcommand{\axname}[1]{\hyperref[ax:#1]{Axiom~#1}}

\title{\textbf{Energy Calculus}\\
       {\large A Compositional Algebra of Energy in Computational Systems}}
       
\author{\normalsize Mosharaf Chowdhury$^{*}$, Jae-Won Chung$^{*}$, Jeff J. Ma, Nishil Talati, Ruofan Wu}

\begin{document}

\begin{mlecover}
\input{sections/abstract}
\end{mlecover}
{\renewcommand{\thefootnote}{}\footnotetext{Authors listed alphabetically. \quad $^{*}$Equal contribution. \quad Correspondence: \texttt{calculus@ml.energy}}}
\clearpage

\setcounter{tocdepth}{2}
\tableofcontents
\clearpage

\input{sections/introduction}
\input{sections/background}
\input{sections/calculus}
\input{sections/reduction}
\input{sections/algebra}
\input{sections/frontier}
\input{sections/discussion}
\input{sections/related}
\input{sections/conclusion}

\clearpage
\appendix
\crefalias{section}{appendix}
\crefalias{subsection}{appendix}
\input{sections/appendices}

\clearpage
{
  \small
  \bibliographystyle{plain}
  \bibliography{refs}
}

\end{document}

%% file: sections/abstract.tex
\begin{abstract}
Energy is a binding constraint for AI scaling, yet it lacks the formal treatment that computation, communication, and learning have long enjoyed.
Recent systems demonstrate large energy savings, but each targets a specific granularity and structure; one cannot combine frequency scaling from one system with critical-path analysis from another and reason about their joint effect on total energy.
Energy remains a monolithic scalar that is measured after the fact and optimized with point solutions that do not generalize.

\medskip

We propose \emph{energy calculus}, a compositional algebra that treats energy as a first-class primitive.
It builds on \emph{energy elements}, units of computation whose energy we can reliably measure, each carrying an \emph{energy signature} that comprises its time, its static and dynamic energy, the hardware operating point and execution context under which we measured it, and the associated measurement uncertainty.
Three operators (sequential, same-device parallel, and cross-device parallel) compose signatures along the same structure as the computation itself, covering arbitrary DAG-structured executions.
The algebra rests on seven axioms that capture how hardware consumes energy, and it exhibits two properties distinctive to energy among computing resources: sequential composition commutes only when elements are mutually context-insensitive, and sequential composition does not distribute over parallel composition.
We also present a Reduction Theorem that recovers simple context-independent algebra whenever interactions fall below measurement uncertainty, so practitioners pay for context dependence only where the physics demands it.
Uncertainty propagates through every composition, so each prediction carries an error bound.
Finally, we show that the same operators extend from energy totals to time--energy Pareto frontiers, so reasoning about tradeoffs composes with the same algebra.
\end{abstract}

%% file: sections/introduction.tex
\section{Introduction}
\label{sec:introduction}

Energy is a binding constraint for AI scaling~\cite{bloombergnef25,cbre2025,inferencemax-blog25}.
While the demand for energy to power millions of accelerators keeps growing, energy procurement at scale remains slow: it takes around three years for natural gas sources and five to ten years for nuclear~\cite{eia-utility-data}.
As such, energy has become a computing resource that deserves the same level of attention and analysis as time, memory, and network bandwidth.

Recent work has advanced our understanding of energy as a quantifiable computing resource and derived large energy savings through optimization.
For instance, Zeus~\cite{zeus:nsdi23} leverages the fact that GPUs are not power-proportional and jointly tunes the batch size of training and the GPU's power limit to cut energy consumption by 24\%--75\%.
Perseus~\cite{perseus:sosp24} identifies ``energy bloat'' in large model training, where viewing training as a Directed Acyclic Graph (DAG) and analyzing the critical path allows up to 30\% reduction in training energy consumption with no slowdown.
Kareus~\cite{kareus:osdi26} decomposes energy into static and dynamic components and jointly optimizes GPU streaming multiprocessor (SM) allocation, computation/communication scheduling, and frequency, reducing training energy by up to 28\% at the same training time.

However, such systems target specific computation granularities and structures, and do not automatically compose or generalize.
One cannot trivially take frequency scaling from Zeus, critical-path analysis from Perseus, and launch timing from Kareus and reason about their combined effect on total energy.
Further, their optimization algorithms are tied to the specific structures they were designed for (e.g., recurring training jobs~\cite{zeus:nsdi23}, the DAG of training forward and backward computations~\cite{perseus:sosp24}, repeating partitions of computation and communication kernels~\cite{kareus:osdi26}).

This calls for a general formalism for reasoning about energy in computational systems.
Mature formalisms for reasoning about computation (complexity theory), communication (information theory), and learning and generalization (statistical learning theory) exist, but energy has yet to be treated as a composable quantity grounded in how hardware consumes it.
It is treated as a monolithic scalar, measured after the fact, and optimized by point solutions that do not generalize across granularity, computation structure, or hardware.
The formalism must decompose energy into algebraically composable components, propagate measurement-induced uncertainty\footnote{Software-reported energy measurements can deviate from hardware ground truth by $\pm$5\%~\cite{nvml}.} through composition, and apply at arbitrary granularity on any hardware.

We propose \emph{energy calculus} to fill this gap.
Energy calculus treats energy as a first-class primitive, a structured quantity decomposed into static and dynamic components that compose along the same structure as the computation itself.
The basic unit of analysis is an \emph{energy element}, defined as any unit of computation for which energy can be reliably measured.
Each element carries an \emph{energy signature}: its measured time and energy, the static--dynamic breakdown, the hardware operating point and execution context under which it was measured, and the associated measurement uncertainty.
By composing per-element energy signatures with explicit algebraic operators, energy calculus produces end-to-end energy predictions with bounded uncertainty, reusing previously characterized signatures rather than profiling each computation in full.

Energy calculus is built on the following principles.
First, the framework captures the physical nature of how computing consumes energy.
The static--dynamic decomposition reflects how hardware actually consumes energy, and we treat direct measurement as ground truth.
Second, measurement reliability sets the finest admissible granularity; above that floor~\cite{kareus:osdi26}, an energy element can be a kernel, layer, block, model, pipeline stage, or an entire job, and we operate at whichever granularity the analysis demands.\footnote{Note that finer decomposition does not reduce profiling time and can increase the number of elements to characterize; the saving comes instead from amortization.
An element characterized once on a target accelerator is composed into many computations without re-profiling, so a signature library amortizes its characterization cost across everything that reuses it.}
Third, we compose elements hierarchically so that elements combine into larger elements using the same operators.
A composition is itself an element with a well-defined energy signature.
Fourth, we make measurement uncertainty intrinsic to the calculus.
Every signature carries the variances of its measured quantities---execution time and total energy---and these variances propagate through every composition, so each prediction comes with an error bound.
Fifth, we admit context dependence: signatures are functions of the physical state in which elements run (thermal, cache, power delivery, etc.), and composition propagates that state from one element to the next.
\Cref{sec:calculus,sec:reduction} formalize these principles as seven axioms (five physical, two methodological).
Furthermore, a Reduction Theorem recovers simple context-independent algebra in the common case where interactions fall below measurement uncertainty, so practitioners pay the cost of context dependence only where it is physically required.

Energy calculus makes energy something practitioners can reason about.
A system designer can ask whether a frequency-scaling decision at the kernel level composes with a critical-path optimization at the pipeline level, and the calculus answers algebraically rather than by re-profiling the combination.
Once signatures are characterized on a target accelerator, predicting energy for a new computation composes existing signatures algebraically with propagated uncertainty, instead of requiring a full optimization sweep.
The same algebra extends from energy totals to time--energy Pareto frontiers, so reasoning about tradeoffs at any granularity uses the same operators as reasoning about totals.
Energy becomes a resource that can be budgeted across stages of a pipeline, allocated to subsystems, and traded against time at a chosen granularity.

\paragraph{Scope.}
This work establishes foundations rather than presenting an evaluation.
We define the algebra, ground it physically, state and prove its core properties, and identify the mechanisms that drive context dependence.
Empirical validation at scale, optimization layers built on the algebra, and energy-aware compilation are the subject of subsequent work; \Cref{sec:open-problems} catalogs these and other open directions.

\paragraph{Contributions.}
We make the following contributions:
\begin{itemize}[nosep]
    \item We introduce \emph{energy calculus}, a compositional framework that treats energy as a first-class primitive with static and dynamic components, parameterized by hardware operating point and execution context, and applicable at any granularity from kernel to job.
    We ground it in seven axioms (five physical, two methodological).

    \item We define three composition operators (sequential, same-device parallel, and cross-device parallel), show that they compose arbitrary DAG-structured schedules, and establish their algebraic properties.
    Two of these properties are distinctive to energy: conditional commutativity of sequential composition under context propagation, and non-distributivity of sequential over parallel composition.

    \item We prove a \emph{Reduction Theorem} that recovers context-independent algebra whenever interacting elements are mutually context-insensitive within measurement uncertainty.

    \item We identify the physical mechanisms behind context dependence (thermal, cache and memory, contention, power delivery) and present an interaction-graph workflow for characterizing non-negligible element pairs and prefixes.
    The inventory is open-ended and extensible without algebraic modification.

    \item We extend the algebra from energy totals to time--energy Pareto frontiers, showing that the three operators induce composition rules on operating frontiers.
    Parallel composition further admits slack-to-energy conversion: an element that finishes earlier than the others can be slowed down to reduce energy consumption, and the calculus provides an analytical framework for reasoning about this.
\end{itemize}

%% file: sections/background.tex
\section{Background}
\label{sec:background}

Our primary setting is AI datacenters, where GPUs dominate the energy of training and inference workloads.
Energy calculus only assumes that device power divides into static and dynamic components, and applies to any accelerator with that structure.

\subsection{Energy in Computational Systems}
\label{sec:energy-systems}

A GPU's power consumption has two components.
\emph{Dynamic power} tracks computation and memory activity: $P_{\text{dynamic}} \propto V^2 f$,\footnote{$V$ is the supply voltage and $f$ is the clock frequency.} dropping to zero when no work occurs.
\emph{Static power} flows continuously whenever the device is powered.
It includes leakage current, memory refresh, clock distribution, and other baseline costs.
For instance, an NVIDIA A100 at idle draws 40--60\,W; at full load, it can reach 300--400\,W depending on SKU and power cap~\cite{nvidia-a100-datasheet}.

The static--dynamic split determines how energy composes.
Two computations running in parallel on one device share static power, drawn once over the makespan.
When the two do not slow each other down by competing for shared resources, their dynamic energies add.
Two computations running sequentially on one device each draw static power for their own duration.
GPUs are not power-proportional: a GPU at 10\% utilization does not consume 10\% of peak power, because static power creates a high floor.
How sub-computations share or serialize on hardware therefore determines whether static power amortizes or goes to waste.

Hardware further shapes these components through runtime controls.
Dynamic voltage and frequency scaling (DVFS) varies $f$ and $V$ at runtime, while a power cap $\Pcap$ enforces an upper bound on instantaneous power draw.
Both static power and dynamic energy depend on $(f, V, \Pcap)$: the same computation at different settings consumes different energy in different amounts of time.
We collect these settings into an \emph{operating point} $\omega = (f, V, \Pcap, \ldots)$ that parameterizes the energy of every element.
These settings are device-wide; an element additionally carries per-element settings such as its SM allocation.
\Cref{sec:energy-elements} splits~$\omega$ into device-wide and per-element components accordingly.

\subsection{Computation as Structured Composition}
\label{sec:composition}

A computation arranges units of any granularity (kernel invocations, fused operator groups, layers, model blocks, pipeline stages, entire jobs) under execution and data dependencies.
That arrangement determines how the units' energies combine into total energy.

The parallelism strategies used in large model training illustrate the typical composition patterns:
\begin{itemize}[nosep]
    \item \emph{Data parallelism} replicates the model across devices.
    Each replica processes a shard of the minibatch (parallel across devices), followed by gradient synchronization (communication overlapping with or following computation).
    \item \emph{Tensor parallelism} partitions individual operators across devices within a layer.
    A single layer's forward pass interleaves local computation with collective communication, creating fine-grained same-device overlap, with cross-device synchronization at the communication boundary.
    \item \emph{Pipeline parallelism} partitions layers across devices into stages.
    A minibatch is split further into microbatches that flow through the pipeline, creating sequential processing within each stage, parallel execution of different microbatches across stages, and idle ``bubble'' time at the boundaries due to pipeline fill and drain, and workload imbalance.
\end{itemize}
In practice these strategies compose: a training job may use tensor parallelism within a node, pipeline parallelism across nodes, and data parallelism across node groups.
The resulting execution is a DAG of computational units connected by sequential, same-device parallel, and cross-device parallel relationships.
Newer regimes such as expert parallelism (for Mixture-of-Experts models) and context/sequence parallelism (for long-context training/inference) decompose into the same sequential and parallel relationships.

The same compositional structure appears beyond training: inference serving, data-processing frameworks, and scientific simulations all arrange sub-computations in time and across hardware through the same sequential and parallel relationships.

\subsection{Measurement Uncertainty}
\label{sec:measurement-uncertainty}

Energy measurement is inherently noisy.
Internal and external instruments measure different quantities at different timescales, and software-reported energy can deviate from hardware ground truth by $\pm$5\%~\cite{nvml}.
Device temperature drifts during measurement, creating run-to-run variance especially for short-duration kernels.
Sampling-interval limits inflate uncertainty for fine-grained measurements.
Different hardware platforms expose different baseline behaviors, complicating cross-hardware comparison further.

Three design requirements follow.
Every measured element must carry a variance term, uncertainty must propagate through every composition, and divergence between a composed prediction and a direct measurement (beyond propagated uncertainty) must flag a composition or characterization failure (\Cref{sec:uncertainty}).

\subsection{Composing the Energy of a Transformer Attention Layer}
\label{sec:attention-example}

We illustrate the core mechanism on a concrete workload before developing the formalism.
The workload is an attention layer from Llama~3.2~3B training on two A100 GPUs with tensor parallelism and computation--communication overlap.

The attention layer consists of a sequential prefix (RMS~Norm $\to$ Linear~1) followed by a parallel region where the compute path (RoPE $\to$ Attention $\to$ Linear~2) overlaps with AllReduce communication (\Cref{fig:parallel-execution}).
The compute kernels use 105 streaming multiprocessors (SMs) while AllReduce uses 3~SMs.

\begin{wrapfigure}{r}{0.5\textwidth}
    \centering
    \includegraphics[width=0.4\textwidth]{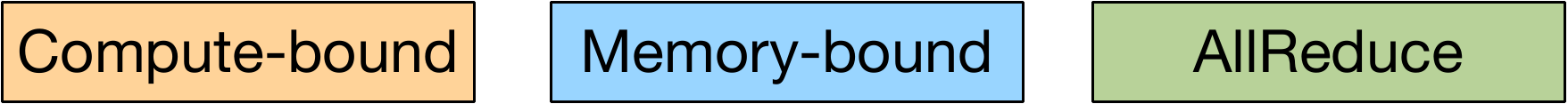}\\[6pt]
    \includegraphics[trim=100pt 0 0 0, clip, width=0.48\textwidth]{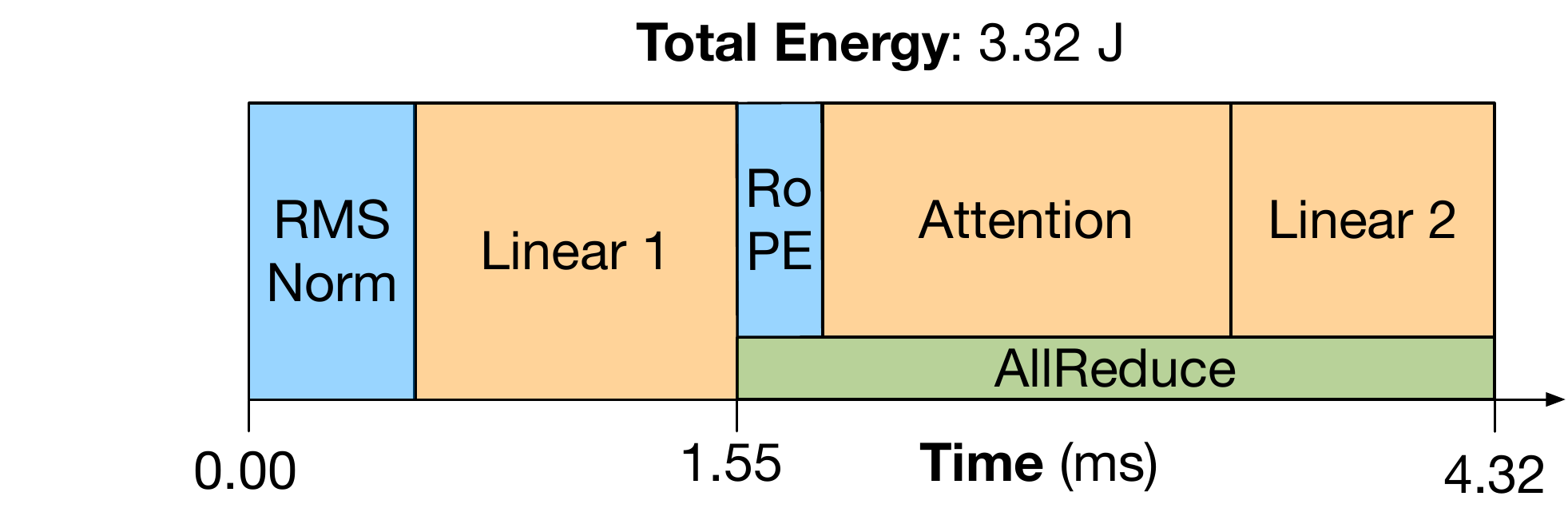}
    \caption{AllReduce overlaps with compute kernels during attention layer execution on 2$\times$A100 GPUs.}
    \label{fig:parallel-execution}
\end{wrapfigure}

\paragraph{Energy signatures.}
Each element has a measurable \emph{energy signature}: a static power~$s$ (W), execution time~$T$ (s), and dynamic energy~$d$ (J), so that total energy is $E = s \cdot T + d$.
For this layer, the sequential prefix takes $T_{\text{seq}} = 1.55$\,ms with dynamic energy $d_{\text{seq}}$;
the subsequent compute path takes $T_{\text{compute}} = 2.77$\,ms with dynamic energy $d_{\text{compute}}$;
and the AllReduce in parallel takes $T_{\text{ar}} = 2.77$\,ms with dynamic energy $d_{\text{ar}}$.
Each A100 draws 60\,W of static power at this operating point, so $s = 120$\,W summed across both GPUs; dynamic energies similarly aggregate across the pair.
Because the two GPUs execute identical per-GPU timelines in lockstep, the cross-device parallelism between them collapses into this combined static power, leaving the within-GPU compute--AllReduce overlap as the only operator structure visible in the formula below.

\paragraph{Composition.}
Given these signatures, the layer's total energy follows from how its parts are arranged in time.
In the parallel region, compute and AllReduce overlap on the same device, so static power is drawn \emph{once} over the region's makespan while dynamic energies add:
\begin{equation}
\label{eq:preview-parallel}
    E_{\text{parallel}} = s \cdot \max(T_{\text{compute}}, T_{\text{ar}}) + d_{\text{compute}} + d_{\text{ar}}
\end{equation}
The full layer composes the prefix and the parallel region sequentially, where each part draws static power for its own duration:
\begin{equation}
\label{eq:preview-layer}
    E_{\text{layer}} = s \cdot T_{\text{seq}} + d_{\text{seq}} + s \cdot \max(T_{\text{compute}}, T_{\text{ar}}) + d_{\text{compute}} + d_{\text{ar}}
\end{equation}
The total time is $T_{\text{layer}} = 1.55 + \max(2.77, 2.77) = 4.32$\,ms.

\paragraph{Static energy accounting.}
Compositional reasoning matters for the static term even when the execution is fixed.
Because the compute path and AllReduce run concurrently, the layer occupies the device for $T_{\text{layer}} = 4.32$\,ms and the correct static energy is $s \cdot T_{\text{layer}} = 120 \times 4.32\text{\,ms} = 518$\,mJ.
Treating the two as if they were serialized would sum their durations to $1.55 + 2.77 + 2.77 = 7.09$\,ms and give $120 \times 7.09\text{\,ms} = 851$\,mJ, overstating static energy by 333\,mJ (64\% of the correct value) by double-counting the overlap period.
Dynamic energy is unchanged in both cases because the same computations execute either way.
Direct measurement of the full layer yields 3.32\,J, of which the composed static term accounts for 518\,mJ; the rest is the three dynamic terms.

This example rests on two simplifying assumptions.
Each element consumes the same time and energy whether measured alone or as part of the layer, so a single measurement characterizes it once and for all.
The overlapping kernels do not slow each other down by competing for shared resources, so each keeps the duration and energy it has on its own.
Later sections show when these assumptions hold and how to handle the cases where they fail.

%% file: sections/calculus.tex
\section{Energy Calculus: Foundations}
\label{sec:calculus}

This section develops the foundational layer of energy calculus.
We define energy elements and the context variables that carry physical state and co-execution structure (\Cref{sec:energy-elements}), build energy signatures over them (\Cref{sec:energy-signature}), and give the composition operators in their general context-dependent form, showing that they compose arbitrary DAG-structured schedules (\Cref{sec:composition-operators}).
Finally, we ground the context variables physically, identifying the interaction mechanisms through which one element's execution changes the energy another consumes (\Cref{sec:mechanisms}).

Based on the foundations established in this section, \Cref{sec:reduction} establishes when context dependence can be dropped, and \Cref{sec:algebra} derives properties and rules of the composition algebra.

\subsection{Energy Elements and Context Variables}
\label{sec:energy-elements}

\begin{definition}[Energy Element]
\label{def:energy-element}
An \textbf{energy element} $e$ is a unit of computation that admits an energy measurement (or estimate) and composes with other elements into larger ones.
Each element is associated with \textbf{workload parameters}~$\theta$ describing the computation it performs, is realized by an \textbf{implementation}~$\kappa$ chosen among those computing its operation, and executes at an \textbf{operating point}~$\omega = (\omega_{\text{dev}}, \omega_{\text{elem}})$ of the hardware, split into a \textbf{device-wide part}~$\omega_{\text{dev}}$ shared by everything on the device and a \textbf{per-element part}~$\omega_{\text{elem}}$ partitioning device capacity among elements.\footnote{Which components are device-wide is hardware-dependent: on current GPUs, frequency, voltage, and the power cap are device-wide, $\omega_{\text{dev}} = (f, V, \Pcap)$, while software knobs that partition the GPU's resources (e.g., SM allocation) are per-element.}
Each element has a \textbf{trace}: a record of the actions it issues, e.g., the cache lines it touches or the memory transactions it requests.
\end{definition}

A GPU kernel, a layer, a Transformer block, a training step, or an entire job are all energy elements; the calculus prescribes no fixed granularity.
The contents of~$\theta$ likewise depend on granularity: tensor dimensions for a kernel, batch size and sequence length for a layer, the full configuration for a training run.
The implementation~$\kappa$ selects among the realizations of the element's mathematical operation.
For instance, the attention operation is one element, and cuDNN, FlexAttention, and FlashAttention are values of~$\kappa$, each with its own trace, time, and dynamic energy.

Beyond what an element mathematically computes ($\theta$), how it is realized ($\kappa$), and how the hardware is set ($\omega$), the energy consumption of an energy element depends on the physical environment of its execution, which we capture with three \emph{Context Variables}.

\begin{definition}[Context Variables]
\label{def:contexts}
The \textbf{entry context}~$\phi$ of an element is the vector of physical state variables it inherits from prior execution, e.g., device temperature, cache occupancy, memory fragmentation, power delivery state.
The \textbf{runtime context}~$\rho$ of an element is the set of elements, including itself, that run concurrently and potentially interfere with each other.
Because each element executes on a specific device, $\rho$ determines the hardware resources the elements share and hence the mechanisms through which they can interfere.
The \textbf{exit context} of an element is the physical state it leaves behind; whatever runs next inherits that as its entry context.
All context variables must be measurable or controllable, as state that can be neither observed nor set cannot be used to explain energy.
\end{definition}

Entry and exit contexts are state handed across time, whereas the runtime context is company kept across space.
Elements \emph{interact} when one's execution impacts another's energy, and the context variables are the complete interface for capturing that influence; a preceding element influences its successors through entry and exit contexts, and concurrently running elements influence each other through the runtime context they share.
Physical mechanisms (e.g., cache, memory, heat) define how interactions via each mechanism translate to actual context variable values.

What ties together the context variables is the element's trace.
The trace itself is fixed given the element, its workload parameters, and its implementation: context changes not what an element does but how long it takes and how much energy it consumes.
The trace records issued actions, not their hardware realization.
Whether a touched cache line hits or misses depends on context, and that difference lands in time and dynamic energy, not in the trace.
Thermal interaction, for example, takes the current the element drew---which, at the device's operating voltage, delivers energy that all ends up as heat released over its execution time---and computes the device temperature left behind.
We survey physical mechanisms and their \emph{context update rules} in \Cref{sec:mechanisms}.

Energy elements compose in their execution in many ways.
In the simplest case where two elements execute back to back on the same device and each runs in \emph{isolation} (its runtime context contains only itself), the first element's exit context is exactly the entry context the second inherits.
Elements can also overlap in time and share a runtime context $\rho$: on one device they take the same entry context and contend for its shared resources (e.g., memory bandwidth), while across devices they each take the entry context of their own device and contend for shared resources (e.g., a common interconnect, node-level power budget).
The traces of all elements in the runtime context contribute influence to the exit context of the whole overlapped execution.

\subsection{Energy Signatures}
\label{sec:energy-signature}

The \emph{energy signature} is the atom of the calculus: every energy element has an energy signature, and composition consumes its operands' energy signatures and produces the composed element's energy signature.
Throughout \Cref{sec:calculus,sec:reduction} we will introduce seven axioms that the calculus rests on, starting from the physical distinction between power that flows regardless of computation and energy that scales with the work performed.

\begin{axiom}[P1: Energy Decomposition]
\label{ax:P1}
Energy over an execution period is the time integral of power; therefore, energy over disjoint periods is additive.
For any energy element~$e$ with workload parameters~$\theta$ and implementation~$\kappa$ at operating point $\omega$ (device-wide part $\omega_{\text{dev}}$), in entry context~$\phi$ and runtime context~$\rho$, the energy consumed decomposes as:
\begin{equation}
\label{eq:energy-signature}
    E_e(\theta, \kappa, \omega, \phi, \rho) = s_e(\omega_{\text{dev}}, \phi, \rho) \cdot T_e(\theta, \kappa, \omega, \phi, \rho) + d_e(\theta, \kappa, \omega, \phi, \rho)
\end{equation}
where $s_e \geq 0$ is the static power of the device $e$ executes on, $T_e \geq 0$ is execution time, and $d_e \geq 0$ is dynamic energy.
\end{axiom}

In the special case of isolation ($\rho = \{e\}$), we suppress $\rho$ and write $E_e(\theta, \kappa, \omega, \phi) = s_e(\omega_{\text{dev}}, \phi)\,T_e(\theta, \kappa, \omega, \phi) + d_e(\theta, \kappa, \omega, \phi)$.
Values with a non-singleton~$\rho$ are whole-run values, absorbing the arrivals and departures of other peer elements running together.
In this case, both who the peers are and when they overlap matter to these values.
The values therefore depend on the \emph{co-scheduled configuration}: the runtime context $\rho$ together with the overlap windows, which the composition supplies since it sets every element's start and end times.
Two schedules that overlap the same elements differently have different characterized values.

Static power\footnote{We define static power to be not the deep idle/sleep power of the device, but rather the power draw of a \emph{ready} device at the same operating point and context as the element's execution. For instance, it would correspond to the NVIDIA GPU Performance state P0.} does not depend on~$\theta$ or~$\kappa$ as it is the baseline draw independent of the workload running on the device, and it depends on the operating point only through its device-wide part $\omega_{\text{dev}}$, since the per-element part $\omega_{\text{elem}}$ is tied to each element's execution decision which is orthogonal to the device's static power draw.
Static power is also a function of the runtime context $\rho$.
For instance, two elements that execute together may increase instant power draw or thermal load that newly or more severely trigger frequency throttling, which affects the device's static power.
Dynamic energy is the derived quantity ($d_e \coloneqq E_e - s_e T_e$), and the axiom asserts $d_e \geq 0$.

\begin{definition}[Energy Signature]
\label{def:energy-signature}
The energy signature of an energy element $e$ is
\begin{equation}
\label{eq:complete-signature}
    S_e = \bigl(e,\; \theta,\; \kappa,\; \omega,\; \Phi,\; \mathcal{P},\; s_e(\cdot),\; T_e(\cdot),\; d_e(\cdot),\; \sigma^2_{T_e},\; \sigma^2_{E_e},\; M\bigr)
\end{equation}
where $\Phi$ is the set of entry contexts and $\mathcal{P}$ the set of co-scheduled configurations (runtime contexts with their schedule windows) over which the element was characterized.\footnote{One can think of the individual entries in $\Phi$ and $\mathcal{P}$ as different environments where $e$ was characterized.}
Additionally, $s_e$, $T_e$, and $d_e$ are the components of \axname{P1} characterized over $\Phi$ and $\mathcal{P}$; $\sigma^2_{T_e}$ and $\sigma^2_{E_e}$ are the variances of the element's execution time and total energy; and $M$ is metadata (measurement method, hardware platform identity, conditions under which the signature was obtained).
\end{definition}

Composing energy elements (\Cref{sec:composition-operators}) is the core operation of energy calculus, and the energy signature is defined to contain every piece of information of the element needed to carry out the composition.
Hardware identity enters through the metadata~$M$: the achievable ranges of~$\omega$ and~$\phi$, and the signature functions themselves, are hardware-specific, so the same abstract computation has different signatures on different hardware, and predicting on other hardware means composing that hardware's own characterized signatures with the same operators.

\begin{axiom}[M1: Intrinsic Uncertainty]
\label{ax:M1}
Error in every measured signature component is summarized by its mean and variance.
For an element measured directly, static power is a calibration constant, exact at each characterized $(\omega_{\text{dev}}, \phi, \rho)$, and the errors of its execution time and dynamic energy are uncorrelated.
Uncertainty propagates through every composition.
\end{axiom}

Uncertainty is part of the signature because a reported energy value without an uncertainty estimate is incomplete.
The signature carries the variances of the two measured quantities, execution time and total energy, and the variance of the derived dynamic energy follows from them.
For an element measured directly, \axname{M1} gives that static power is exact and the errors of time and dynamic energy are uncorrelated, so \axname{P1} yields:
\begin{equation}
\label{eq:energy-variance}
    \sigma^2_{E_e} = s_e^2\,\sigma^2_{T_e} + \sigma^2_{d_e}.
\end{equation}
A composed element carries its energy variance directly, propagated from its parts.
The variances carry the random error: instrument noise and run-to-run jitter.
Heavier-tailed or multimodal error structure falls outside the model.
The three components are not measured symmetrically: execution time carries run-to-run variability, static power is a calibration constant and contributes no variance, and dynamic energy, derived per run as $E_e - s_e T_e$ (\axname{P1}), inherits the energy instrument's noise.\footnote{Time can often be measured with higher precision than energy in modern hardware, so $\sigma^2_{E_e}$ would be dominated by energy variability.}
If a run that happens to take $\Delta T$ longer contributes $s_e\,\Delta T$ more to its energy fluctuation, the subtraction $d_e = E_e - s_e T_e$ cancels that component. 
This leaves dynamic-energy error largely orthogonal to time, supporting \axname{M1}'s uncorrelated assumption.

\subsection{Composing Energy Elements}
\label{sec:composition-operators}

We define three composition operators that combine energy elements: sequential ($\seqop$), same-device parallel ($\parop$), and cross-device parallel ($\parml$).
The operators accept \emph{energy signatures} as input and produce the energy signature of the composite element.
For each operator, we derive the composite's energy consumption $E_{e_1 \,\mathrm{op}\, e_2}$, abbreviated $E_1 \,\mathrm{op}\, E_2$; \Cref{prop:closure} later proves the operators closed, deriving the composite's own signature components.
For notational convenience, $\omega$ and each element's $\theta$ and $\kappa$ are suppressed.

\paragraph{Idle elements and energy accounting.}
We write $\mathcal{D}_e$ for the set of devices hosting element $e$: a singleton for an element on one device, and larger for composites spanning several.
When multiple elements compose, all participating devices across the composing elements are held for the full makespan of the composition.
Any period in which a participating device runs nothing is viewed as running a special \emph{idle element}: an energy element with an empty trace, whose device stays ready at its element's device-wide operating point, drawing that device's static power at the idle period's own entry context and zero dynamic energy.
Thus, each participating device consumes static power for the full makespan of the composition (either via an element or an idle element), and the dynamic energy consumption of all elements adds on top of that static energy consumption.
We note that idle elements appear in runtime contexts like any other element, since an idle device still draws power and dissipates heat on shared rails and cooling paths, coupling it and concurrent elements in both directions.

\subsubsection{Sequential Composition ($\seqop$)}
\label{sec:sequential-composition}

For two elements executing one after the other on the same device, the composite's energy is the sum of the two executions' energies (\axname{P1}), each evaluated at its own entry and runtime contexts:
\begin{equation}
\label{eq:sequential}
    E_1 \seqop E_2 = E_{e_1 \seqop e_2}(\phi_0) = s_1(\phi_0, \rho_1)\,T_1(\phi_0, \rho_1) + d_1(\phi_0, \rho_1) + s_2(\phi_1, \rho_2)\,T_2(\phi_1, \rho_2) + d_2(\phi_1, \rho_2),
\end{equation}
where $\phi_0$ is the entry context of the first element and the composite, $\phi_1$ is $e_1$'s exit context, inherited by $e_2$ as its entry context, and $\rho_i$ is the runtime context the surrounding composition places $e_i$ in---the two elements run back to back, so neither appears in the other's.
For $n$ operands, the composite's energy is still the sum of the $n$ elements' total energy consumption, each evaluated at its own entry context, which is the exit context of its predecessor, and its own runtime context.

\subsubsection{Same-Device Parallel Composition ($\parop$)}
\label{sec:same-device-parallel}

Working out the energy consumption of elements that start together and overlap inside one device starts from a physical fact: however many elements run, the device draws its static power once.

\begin{axiom}[P2: Per-Device Static Power]
\label{ax:P2}
A device has a single static power, reflecting components that draw power regardless of the workload (e.g., leakage, clock distribution, memory refresh), so however many elements run concurrently in the device, the device draws static power once.
\end{axiom}

Following \axname{P2}, the composite's total energy consumption is:
\begin{equation}
\label{eq:same-device-parallel}
    E_1 \parop E_2 = E_{e_1 \parop e_2}(\phi_0) = s(\phi_0, \rho) \cdot \max\bigl(T_1(\phi_0, \rho), T_2(\phi_0, \rho)\bigr) + d_1(\phi_0, \rho) + d_2(\phi_0, \rho)
\end{equation}
where $\phi_0$ is the shared entry context of both elements, $\rho$ is the runtime context the surrounding composition places the pair in ($\{e_1, e_2\}$ when the pair runs alone), and $s$ is the device's static power ($s = s_1 = s_2$ as they run on the same device).
Following the same energy accounting principle, the device draws its static power over the whole makespan of the composition ($\max(T_1, T_2)$), and that static power is drawn once per device (\axname{P2}).
Compared with each element executing in isolation, concurrently running elements can \emph{interact} with each other (\Cref{sec:mechanisms}), slowing each other down and perturbing each other's dynamic energy---this is what the runtime context $\rho$ captures.
For $n$ elements running on the same device together, the makespan is $\max_i T_i$ over which the device's static power is consumed, and dynamic energies add.

The same-device parallel operator $\parop$ is for all elements that start executing at the same time.
Cases where elements start at different times are captured by sequential composition with idle elements.
That is, such a schedule's energy consumption can be computed by explicitly introducing idle elements in front of the later-starting elements, sequentially composing them with $\seqop$, and then composing the resulting elements with $\parop$.

\subsubsection{Cross-Device Parallel Composition ($\parml$)}
\label{sec:cross-device-parallel}

For two elements executing concurrently on \emph{different} devices for the same duration, each device draws its own static power:
\begin{equation}
\label{eq:cross-device-parallel}
    E_1 \parml E_2 = E_{e_1 \parml e_2}(\phi_0) = s_1(\phi_0, \rho)\,T_1(\phi_0, \rho) + d_1(\phi_0, \rho) + s_2(\phi_0, \rho)\,T_2(\phi_0, \rho) + d_2(\phi_0, \rho),
\end{equation}
where $\phi_0$ is the composite entry context each device starts from, $\rho$ is the runtime context the surrounding composition places the pair in ($\{e_1, e_2\}$ when the pair runs alone), and $T_1 = T_2$ by the operator's equal-duration requirement.
Runtime context $\rho$ matters across devices too: concurrent elements running on different devices can still interfere (e.g., PCIe switch, NVLink/NVSwitch/NVL72 domain, a network link, node- or rack-level power budgets) if they share and saturate a resource.
For $n$ elements of equal duration, each device's static power is drawn over the common duration, and dynamic energies add.

The cross-device parallel operator $\parml$ is for all elements that start executing at the same time and run for the same duration.
Cases where the start and end times differ across elements are captured by sequential composition with idle elements, similarly to the same-device parallel case.
Specifically, such a schedule's energy consumption can be computed by explicitly introducing idle elements in front of the later-starting elements and after the earlier-ending elements, sequentially composing them with $\seqop$, and then composing the resulting elements with $\parml$.

\begin{figure*}[t]
  \centering
  \subfloat[The execution plan as a DAG.]{
    \includegraphics[scale=0.45]{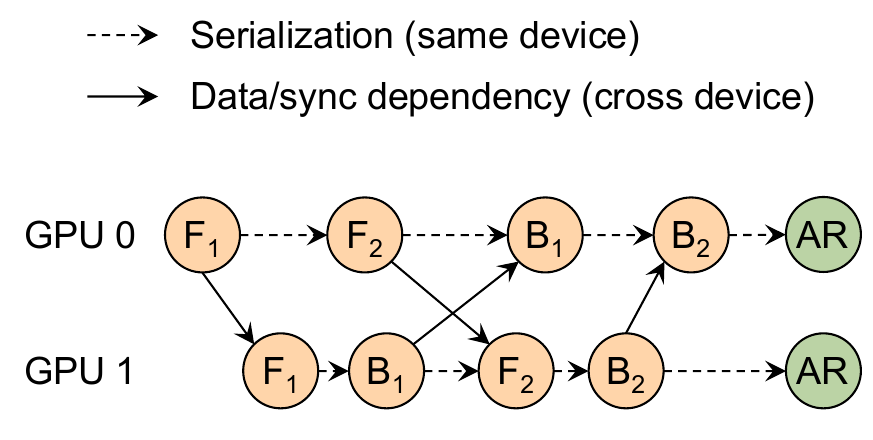}
    \label{fig:composition-case-study-dag}
  }
  \hfil
  \subfloat[The plan reduced to per-GPU sequences with idle elements.]{
    \includegraphics[scale=0.45]{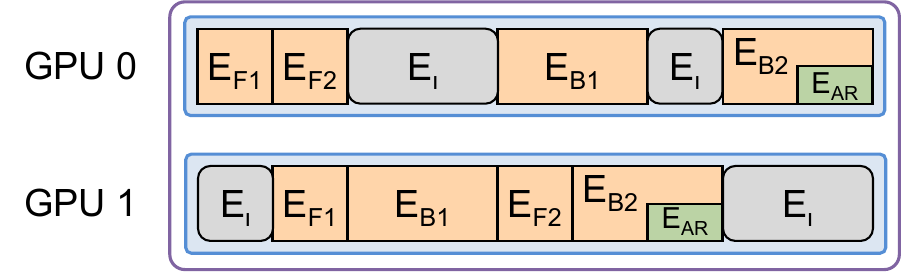}
    \label{fig:composition-case-study-reduced}
  }
  \caption{
    Composing energy of a DAG-structured schedule with the three operators.
    (a) A two-stage 1F1B pipeline-parallel execution of two microbatches, running as one replica of a data-parallel job; $\mathit{AR}$ is the gradient AllReduce with the peer replica (not shown).
    (b) Idle elements cover every gap and pad the later-starting $\mathit{AR}$, which $\parop$ merges with $B_2$.
    Each blue box is the element $\seqop$ composes within its GPU, and the purple box around all GPUs is the final $\parml$ element.
  }
  \label{fig:composition-case-study}
\end{figure*}

\subsubsection{Composition Over a DAG of Elements}
\label{sec:dag-composition}

So far, we have introduced each operator over exactly two elements, with one device in the world for $\seqop$ and $\parop$ and two for $\parml$.
The key challenge is to determine how entry contexts, runtime contexts, and schedules are derived for each element under arbitrary compositions.
Real execution schedules (e.g., pipeline parallelism, data parallelism, complex overlap patterns) form DAGs of elements, and this section composes any such schedule with the three operators.

We take a DAG to be a fully scheduled execution plan, not a bare dependency specification: each element carries its placement $\mathcal{D}_e$, and the edges are the data and synchronization dependencies plus the serialization edges between elements executing on the same device.
The schedule determines the contexts given to every element.
Only data and synchronization dependencies cross the device boundary, not hardware state, so an element's entry context is its own device's state when it starts.

Any such plan can be expressed as a term over the three operators, with \Cref{fig:composition-case-study} providing an example.
The general procedure is as follows.
First, merge the elements that overlap on one device into a single element via $\parop$, prepending idle elements to later-starting ones if needed.
Second, insert an idle element into every period in which a device runs nothing, including the gaps before a device's first execution and after its last, so that every device is covered for the full makespan.
Finally, compose each device's elements with $\seqop$ in execution order; the resulting per-device elements all span the same makespan, the length of the critical path through the plan, so $\parml$'s equal-duration requirement is met by construction and composes them across devices.
In \Cref{fig:composition-case-study}, the gradient AllReduce starts partway through $B_2$ on each GPU, so the first step pads it with an idle element and $\parop$ merges the pair; the full plan becomes:
\begin{equation}
E =
\Bigl(
E_{F_1} \seqop E_{F_2} \seqop E_{\iota} \seqop E_{B_1} \seqop E_{\iota} \seqop \bigl(E_{B_2} \parop (E_{\iota} \seqop E_{\mathit{AR}})\bigr)
\Bigr)
\parml
\Bigl(
E_{\iota} \seqop E_{F_1} \seqop E_{B_1} \seqop E_{F_2} \seqop \bigl(E_{B_2} \parop (E_{\iota} \seqop E_{\mathit{AR}})\bigr) \seqop E_{\iota}
\Bigr).
\end{equation}
This process is exact: every instant on every participating device belongs to exactly one execution---an element, a same-device parallel group, or an idle element---and energy over disjoint periods adds.
The term is thus the execution schedule in operator form, designating the entry and runtime contexts for every element for the purpose of evaluating the energy signature of composites at each level.

\subsection{Interaction Mechanisms}
\label{sec:mechanisms}

Energy elements can interact with each other and influence each other's energy through physical mechanisms.

\begin{axiom}[P3: Interaction Structure]
\label{ax:P3}
Interaction effects decompose into a finite set of physical mechanisms, each summarizable by a small number of effective state variables.
\end{axiom}

\axname{P3} is an assumption about the structure of interactions: the channels through which one element's execution can influence another's energy are finite and well-defined, which allows us to characterize them and compose them in the calculus, exactly within this model.

Thermal coupling, cache and memory state, and power delivery act through the entry context: they are state accumulated from prior execution, and their effective state variables constitute~$\phi$.
Resource contention acts through the runtime context: it exists only during overlapped execution, is a property of concurrent execution rather than of any entry state, and is characterized at the pair or composite level.
The contended resource can be device-local for $\parop$ (memory bandwidth, cache capacity) or shared across devices for $\parml$ (PCIe switch, NVLink/NVSwitch/NVL72 domain, TCP/RDMA link, node- or rack-level power budgets).

A mechanism acting through the entry context additionally carries a context update rule that advances its effective state variables under the traces present.
Temperature responses add approximately under a linear thermal model; cache occupancy follows a displacement model, traces evicting entry-state residency and one another under the capacity constraint; power-delivery state decays to the trace's ending draw within microseconds.
Exit contexts assemble by these rules: a successor's entry context follows from its predecessor's entry context and trace.
A few effective state variables suffice per mechanism, even when the underlying physical state (e.g., cache contents) is high-dimensional.

Characterization is performed via measurement, and measurement needs a starting state that every run can reproduce: the \emph{null context}~$\phi_{\text{null}}$, the reference entry context the characterization protocol realizes (\axname{M2}), in which no mechanism carries appreciable state from prior execution.

\begin{axiom}[M2: Interaction Measurability]
\label{ax:M2}
Signatures are characterized by direct measurement of the element under multiple representative entry contexts spanning the operating range, with other variables fixed.
Characterization starts from the null context~$\phi_{\text{null}}$, measures the signature components at each context and the element's trace once, and requires a protocol that isolates~$\phi$ from other sources of variation (thermal stability, consistent initial cache state, controlled concurrency).
The runtime context is covered the same way, by co-executing the element with representative runtime contexts.
\end{axiom}

\Cref{tab:interaction-categories} identifies the current inventory and how each mechanism is characterized, and is open to extension.
To add a mechanism, one must specify its physical cause, its characteristic timescale, a measurement protocol that isolates it from other mechanisms, a set of effective state variables, and rules for taking the element's trace and contexts and producing the exit context.
Adding a new mechanism does not change the structure of the calculus: it adds a new entry context variable or a new shared resource acting through the runtime context, and the composition operators remain unchanged.
\Cref{app:interactions} details current protocols.

\begin{table}[t]
\centering
\caption{Current interaction mechanism inventory, grouped by the context channel through which each mechanism acts.
Full characterization protocols appear in \Cref{app:interactions}.}
\label{tab:interaction-categories}
\begin{tabular}{@{}l >{\raggedright\arraybackslash}p{1.95in} >{\raggedright\arraybackslash}p{1in} >{\raggedright\arraybackslash}p{1.65in}@{}}
\toprule
\textbf{Mechanism} & \textbf{Physical cause} & \textbf{Timescale} & \textbf{Characterization} \\
\midrule
\multicolumn{4}{@{}l}{\emph{Mechanisms acting through the entry context~$\phi$}} \\
\addlinespace
Thermal coupling & Heating increases leakage and can trigger throttling & Seconds & Measure $E$ at multiple device temperatures \\
\addlinespace
Cache/memory state & Cache warming or pollution by preceding elements & Intervening traffic & Measure $E$ with cold vs.\ warm cache \\
\addlinespace
Power delivery & Voltage droops from current-draw transitions & Nanoseconds to microseconds & High-resolution power measurement at boundaries \\
\midrule
\multicolumn{4}{@{}l}{\emph{Mechanisms acting through the runtime context~$\rho$}} \\
\addlinespace
Resource contention & Shared resources (bandwidth, cache, fabrics, power budgets) degraded under concurrency & Instantaneous & Measure $E$ in isolation vs.\ concurrently with the interferer \\
\bottomrule
\end{tabular}
\end{table}

%% file: sections/reduction.tex
\section{Reduction to Context-Independence}
\label{sec:reduction}

The generic framework of \Cref{sec:calculus} is complete but can become unwieldy: every signature value is a function of the entry and runtime contexts, and characterizing that dependence mechanism by mechanism becomes less tractable as the number of elements grows.
However, in practice, the full generality may not always be necessary.
Elements that do not share an interaction mechanism's hardware channel do not interact through it at all.
Furthermore, interaction would be limited between elements that are far apart, where distance is mechanism-specific: elapsed time for thermal state and power-delivery transients, and intervening memory traffic for cache state.
In this section, we formalize these observations (\Cref{sec:locality}), present the Reduction Theorem which unlocks a simplified \emph{context-independent} algebra (\Cref{sec:reduction-theorem}), and present the interaction graph, which records characterized interactions between context-dependent elements and lets everything else compose context-independently (\Cref{sec:interaction-graph}).
Finally, we summarize the foundational axioms and notation of context-dependent and context-independent energy calculus (\Cref{sec:axiom-summary}).

\subsection{Context-Insensitivity}
\label{sec:locality}

\subsubsection{Sensitivity to Entry Context}
Certain pairs of elements may interact negligibly.
For entry context dependence, \axname{P4} provides an analytical model for interactions between elements and the decay thereof.

\begin{axiom}[P4: Interaction Locality]
\label{ax:P4}
Entry context effects decay with separation between elements.
Separation is mechanism-specific.
For energy elements $e_1$ and $e_2$, each mechanism~$k$ carries its own separation measure~$\ell_k(e_1, e_2)$: a mechanism couples the two only when they share the hardware it acts through ($\ell_k(e_1, e_2) = \infty$ otherwise), and, given shared hardware, $\ell_k(e_1, e_2)$ measures their separation in the mechanism's own metric.
Then, for each signature component $x \in \{s, T, d\}$:
\begin{equation}
\label{eq:locality}
\bigl|\,x_2(\phi_1) - x_2(\phi_{\lnot 1})\,\bigr| \;\to\; 0 \quad \text{as } \ell_k(e_1, e_2) \to \infty \text{ for all } k
\end{equation}
where $\phi_1$ is the entry context $e_2$ sees following $e_1$'s execution from $\phi_0$, and $\phi_{\lnot 1}$ the entry context $e_2$ would see had $e_1$ not executed, all else unchanged.
\end{axiom}

The two contexts differ only in the state $e_1$ left behind, which decays mechanism by mechanism as the separation grows, and the \emph{interaction effect} of $e_1$ on each of $e_2$'s signature components---hence on its energy consumption---vanishes in the limit.

How $\ell_k$ is defined depends on the interaction mechanism $k$.
For the current inventory (\Cref{tab:interaction-categories}), thermal state couples elements that share a die and cooling path, with $\ell_k$ the elapsed time between them; cache state couples elements that share a cache hierarchy, with $\ell_k$ the memory transactions between the two executions; and power-delivery transients couple elements on the same delivery rail, with $\ell_k$ the elapsed time since the current transition.
The rate of decay as $\ell_k \to \infty$ is mechanism-specific as well, e.g., thermal effects would be slow, whereas with sufficient memory transactions, cache state interactions would decay quickly.

Now, we define when we can treat an element to be insensitive to different entry contexts it executes under, which relieves the costly characterization of the element under every possible context.

\begin{definition}[History-Insensitivity]
\label{def:history-insensitive}
An element~$e$ is \emph{history-insensitive} with respect to a set of entry contexts (i.e., history) within measurement uncertainty if, for all $\phi, \phi'$ in the set,
\begin{equation}
\label{eq:history-insensitive}
    \Delta s\, T_e(\phi') + \Delta T\, s_e(\phi') + \Delta s\, \Delta T + \Delta d \;\leq\; c_\alpha \cdot \sigma_{E_e},
\end{equation}
where $\Delta s = |s_e(\phi) - s_e(\phi')|$, $\Delta T = |T_e(\phi) - T_e(\phi')|$, and $\Delta d = |d_e(\phi) - d_e(\phi')|$, for the desired confidence level~$(1 - \alpha)$, with $c_\alpha$ the two-sided Gaussian multiplier, the $(1 - \alpha/2)$ quantile of the standard normal distribution.
\end{definition}

The left-hand side comes from expanding the absolute total energy difference $|E_e(\phi) - E_e(\phi')|$ into static power, time, and dynamic energy components per \axname{P1} and extracting a worst-case bound containing each component's absolute difference (derivation in \Cref{app:expansion}).
Each difference is measurable: time per run, static power as the ready device's draw at each context, and dynamic energy derived per run as $E_e - s_e T_e$ (\axname{P1}).
When this worst-case bound is satisfied, each component's contribution is bounded separately, and total energy difference is bounded by $c_\alpha \cdot \sigma_{E_e}$ as well.\footnote{Directly bounding the total energy difference is not what we want. For instance, an element may have the same total energy under entry contexts $\phi$ and $\phi'$, but $\phi'$ may lead to longer execution time and lower dynamic energy than $\phi$ by the compensating amount. In this case, a total-energy-only bound test would indicate history-insensitivity, but wherever time stands alone in composition (e.g., in the makespan and the static power it multiplies), the composite's energy under $\phi$ and $\phi'$ would differ significantly. This contradicts history-insensitivity.}

\begin{definition}[Mutual History-Insensitivity]
\label{def:mutual-history-insensitive}
Two elements are \emph{mutually history-insensitive} when each is history-insensitive with respect to the contexts the other can produce.
\end{definition}

History-insensitivity is a property of an element paired with a set of entry contexts, not an intrinsic property of the element.
An element may be history-insensitive with respect to thermal variation within a narrow band and history-sensitive with respect to cache state.
It is relative to measurement uncertainty as well: the threshold $c_\alpha \cdot \sigma_{E_e}$ treats shifts smaller than a chosen multiple of the element's run-to-run variation, which includes mainly measurement instrument uncertainty (\axname{M1}), as negligible.
As instruments improve, the instrument uncertainty component of $\sigma_{E_e}$ shrinks.
Thus, with higher fidelity measurement instruments, the insensitivity check is able to detect interactions for characterization that coarser measurement could not separate from noise.

\subsubsection{Sensitivity to Runtime Context}
Locality and history-insensitivity covers the entry context: it bounds how an element responds to the state its predecessors leave behind.
Concurrent composition needs a second condition, for the runtime context.
For runtime context dependence, \axname{P5} states the physical regularity: unlike entry-context effects, which decay with separation, interference is thresholded.

\begin{axiom}[P5: Interference Thresholding]
\label{ax:P5}
Runtime context effects are thresholded by shared-resource capacity.
Concurrent elements interfere only through the hardware resources they share, each with finite capacity.
While the elements' combined demand on every shared resource stays within its capacity, the interference effect on each element's static power, execution time, and dynamic energy remains negligible; once any shared resource saturates, the effect grows abruptly with the excess demand.
\end{axiom}

\axname{P5} allows us to formalize when elements that run together are insensitive to each other's presence.

\begin{definition}[Interference-Insensitivity]
\label{def:interference-insensitive}
An element~$e$ is \emph{interference-insensitive} with respect to a set of runtime contexts containing it and a set of entry contexts, within measurement uncertainty, if for every $\rho$ and every $\phi$ in the two sets,
\begin{equation}
\label{eq:interference-insensitive}
    \Delta s\, T_e + \Delta T\, s_e + \Delta s\, \Delta T + \Delta d \;\leq\; c_\alpha \cdot \sigma_{E_e},
\end{equation}
where $\Delta x = |x_e(\phi, \rho) - x_e(\phi, \{e\})|$ for $x \in \{s, T, d\}$, and the multipliers $T_e$, $s_e$ are the isolation values $T_e(\phi, \{e\})$, $s_e(\phi, \{e\})$.
\end{definition}

The left-hand side mirrors \Cref{eq:history-insensitive}: it is the same \axname{P1} expansion applied to the co-run and isolation runs, $|E_e(\phi, \rho) - E_e(\phi, \{e\})|$, so satisfying it bounds each component's co-run perturbation separately and the total energy difference as well (\Cref{app:expansion}).

\begin{definition}[Mutual Interference-Insensitivity]
\label{def:mutual-interference-insensitive}
Two elements are \emph{mutually interference-insensitive} when both are interference-insensitive with respect to their shared runtime context $\{e_1, e_2\}$ and the entry contexts they execute under.
\end{definition}

Interference-insensitivity is the runtime counterpart of history-insensitivity.
When elements that run together collectively exceed a shared resource's capacity (\axname{P5}), they ought to be characterized as a composite element, which the interference-insensitivity test detects.

\subsubsection{Putting the Two Together}
A composition realizes both kinds of context, so we put them together in a single definition.

\begin{definition}[Mutual Context-Insensitivity]
\label{def:mutual-context-insensitive}
Two composed elements are \emph{mutually context-insensitive} when each is history-insensitive with respect to the null context and the entry contexts the composition realizes for it, and interference-insensitive with respect to the runtime contexts the composition places it in, at those same entry contexts.
\end{definition}

The definition covers idle elements like any other element: an idle element's duration is set by the schedule and its dynamic energy is zero, so both conditions reduce to their $\Delta s$ terms, with its energy variance inherited from the durations that set its gap (\Cref{sec:composition-operators}).
When a composition is purely sequential and runs alone, every runtime context is a singleton, the differences in \Cref{eq:interference-insensitive} vanish, and the condition reduces to mutual history-insensitivity; the interference half constrains compositions that co-run elements.

\subsection{The Reduction Theorem}
\label{sec:reduction-theorem}

Pairs of elements can be mutually context-insensitive, as formalized by \Cref{sec:locality}.
This allows us to reduce the generic, context-dependent algebra of \Cref{sec:calculus} to a simpler, context-independent form.

\begin{proposition}[Reduction Theorem]
\label{prop:reduction}
Let $e_1$ and $e_2$ be composed by an operator $\mathrm{op} \in \{\seqop, \parop, \parml\}$, with the pair mutually context-insensitive at confidence level~$(1 - \alpha)$.
Write $\mathrm{op}^{(0)}$ for the same composition computed with every signature value evaluated at the null context in isolation.
Then, whatever entry context the composition itself starts at, $E_1 \,\mathrm{op}\, E_2$ and $E_1 \,\mathrm{op}^{(0)}\, E_2$ agree within operator-specific multiples of the elements' \emph{bias tolerances} $c_\alpha \cdot \sigma_{E_{e_i}}$:
\begin{equation}
    \label{eq:reduction}
    \small
    \bigl| E_1 \,\mathrm{op}\, E_2 - E_1 \,\mathrm{op}^{(0)}\, E_2 \bigr| \;\leq\; m_1 \, c_\alpha\, \sigma_{E_{e_1}} + m_2 \, c_\alpha\, \sigma_{E_{e_2}},
    \qquad
    \begin{array}{l@{\;\;}r@{\;\;}r}
        \mathrm{op} & m_1 & m_2 \\
        \midrule
        \seqop & 1 & 1 \\
        \parop & 2 & 2 \\
        \parml & 2 & 2
    \end{array}
\end{equation}
The $\parml$ row covers its equal-duration domain; unequal durations compose through idle padding, whose added terms \Cref{app:reduction-accumulation} derives.
\end{proposition}

\begin{proof}[Proof sketch]
Reduction replaces each element's realized values with null-context isolation values through two kinds of substitution, each shifting the element's static power, time, and dynamic energy by at most one bias tolerance.
Entry-state substitution evaluates the element at the null context instead of its realized entry context, bounded by history-insensitivity; isolation substitution evaluates it in isolation instead of at the shared runtime context, bounded by interference-insensitivity.
The multipliers count the substitutions each element needs: sequential composition runs alone, so each element carries only its entry-state substitution, and the parallel operators co-run their elements, so each carries both.
\Cref{app:reduction} proves each case.
\end{proof}

When the composition itself starts at the null context, elements that start with it skip the entry-state substitution and drop one tolerance each, e.g., $m_1 = 0$ under $\seqop$.
The Reduction Theorem allows us to compose signatures at the null context in isolation, without characterizing every possible entry and runtime context, significantly reducing the characterization burden.
\Cref{sec:algebra} derives properties and rules of composition in both the general algebra and the resulting \emph{context-independent algebra}.

Reduction is widely applicable in practice.
Elapsed time or elements in between separate most pairs enough that their mutual entry context effects fall below measurement uncertainty.
Concurrently running elements are often scheduled carefully to avoid bottlenecking on shared resources (lest they risk little gain or even slowdown), so their interference is often negligible as well (\axname{P5}).
Finally, many workloads run in steady-state thermal and cache regimes where signatures characterized once apply across many compositions; the set of entry and runtime contexts that introduces perturbation larger than the bias tolerance is not large, and the theorem would be applicable to many pairs.
Where these conditions fail (cold-start transients, adjacent thermally-coupled kernels, shared-cache contention), the reduction does not apply, and more exact characterization is needed (\Cref{sec:interaction-graph}).

Composing elements at the null context in isolation (i.e., using the context-independent $\mathrm{op}^{(0)}$ operators) accumulates the bias term: each element's contribution to the composite's energy may be off by up to $m_i$ times its bias tolerance (\Cref{eq:reduction}).
Moreover, for a given composition, the offset is the same on every run, unlike measurement noise which varies run to run.\footnote{The bias has a physical cause that the composition reproduces on every run. For instance, if element $e_2$ always runs right after a high-power element $e_1$, $e_2$ enters at the same elevated temperature on every run and consumes the same extra energy relative to its signature values at the null context. Averaging repeated runs therefore shrinks measurement noise but leaves the bias untouched.}
Thus, over an $n$-element composition, each element contributes its own bias, growing the worst case total bias linearly with $n$.\footnote{Reaching that worst case requires every element's bias to point in the same direction, i.e., sustained context drift like monotonic heating.}
The worst case bound is computable in advance by summing every element's $m_i \, c_\alpha \, \sigma_{E_{e_i}}$ bounds, and this can be checked against what the application tolerates at the planned composition depth (more details in \Cref{app:reduction-accumulation}).
When the bound exceeds that, the remedy is to characterize at a coarser granularity.
Directly measuring a larger composite as a single element replaces every bias term accumulated inside it with one measured signature, at its own measured uncertainty.
However, the tradeoff of coarser grained composition is reduced reusability: the composite's signature is only valid for the contexts it was characterized under, whereas the individual elements' signatures are valid for any composition that does not violate their mutual context-insensitivity.

\subsection{The Context-Independent Form}
\label{sec:ci-form}

When the Reduction Theorem (\Cref{prop:reduction}) applies, we opt into context-independent algebra, where every signature value is evaluated in isolation at the null context $\phi_{\text{null}}$, both~$\phi$ and~$\rho$ are dropped, and the bounded bias (\Cref{eq:reduction}) is accepted.
The signature reduces to point values, $S_e = (e, \theta, \kappa, \omega, \phi_{\text{null}}, s_e,\allowbreak T_e, d_e, \sigma^2_{T_e}, \sigma^2_{E_e}, M)$.
The semantics of the operators themselves do not change: the accounting of \Cref{sec:composition-operators} applies verbatim over the context-independent signatures.

\subsection{The Interaction Graph}
\label{sec:interaction-graph}

Pairs that fail mutual context-insensitivity fall outside the Reduction Theorem, and their interaction must be characterized. 
Locality keeps this set sparse for separated pairs (\axname{P4}), and capacity headroom for concurrent ones (\axname{P5}).
However, a failing pair does not force the whole composition into context-dependent treatment.
Once characterized, the pair composes with its signature evaluated at the contexts the composition realizes; every other pair composes at the null context, and only the latter contribute bias terms.
The characterized pair must in turn be interference-insensitive with respect to the remaining co-runners, with its co-run values in place of the isolation values in \Cref{eq:interference-insensitive}.

\paragraph{The interaction graph.}
Characterized interactions are recorded in the interaction graph.
Nodes are element \emph{types} rather than instances; for example, an attention layer that runs thousands of times per training step is one node, not thousands.
Edges connect pairs of types whose interaction has been characterized. 
There are two kinds of edges, corresponding to the entry and runtime contexts.
\begin{itemize}[nosep]
    \item \emph{Directed edge} ($A \to B$): an entry context interaction, where the state $A$ leaves behind perturbs $B$'s signature.
    \item \emph{Undirected edge} ($A \leftrightarrow B$): a runtime context interaction under co-execution.
        The edge carries each endpoint's perturbation separately, since interference can affect one element of a pair and not the other.
\end{itemize}
Because nodes are types, one edge covers all instances of the two types it connects, and characterization scales with the number of distinct types rather than the length of the execution.
Every edge carries the responsible mechanisms from \Cref{tab:interaction-categories} and the workload parameters and contexts over which the interaction was characterized, analogous to $\Phi$ and $\mathcal{P}$ in the signature.
The edge also references the characterization outcome of step 4 below, which is either the signature extended over the realized contexts or the composite pair's signature.
The measured deviation that triggered characterization is kept as edge metadata, but composition never consumes it.
Composition consumes only signatures, and the interaction graph is a sparse index into them.

Pairwise edges do not by themselves justify the reduction for a whole composition.
Individually negligible effects can accumulate.
Kernels in a long prefix can each heat the device by less than the pairwise bias tolerance while their aggregate shifts a later element past throttling, and any two of three concurrent elements can fit within a shared bandwidth that all three together exceed.
The applicability condition is therefore stated at the level of whole prefixes and sub-DAGs: for every element, its actual entry context under the composed schedule must lie in a characterized set over which it is history-insensitive, and its runtime context must lie in a characterized set over which it is interference-insensitive.
When either fails, the offending prefix becomes a composite element with its own signature and trace, and cumulative effects surface when the composite's measured energy deviates from its composed prediction (\Cref{sec:uncertainty}); the detected cause is then recorded in the interaction graph.

\paragraph{Building the interaction graph.}
The interaction graph starts empty and grows only as measured deviations justify characterization.

\begin{enumerate}[nosep]
    \item \textbf{Context-independent composition.} Compose signatures with every value at the null context in isolation, correct within the Reduction Theorem's tolerances for pairs satisfying its condition.

    \item \textbf{Deviation detection.} Compare composed predictions against direct measurement (\Cref{sec:uncertainty}).
        Pairs with consistent deviation are candidates for the interaction graph.

    \item \textbf{Mechanism identification.} For each candidate pair, identify the dominant mechanism from \Cref{tab:interaction-categories} and measure the interaction effect via the mechanism's characterization protocol.

    \item \textbf{Signature update.} Either characterize the element as a function of~$\phi$ or~$\rho$ (for reusable dependence) or characterize the composite pair as a single element (for pair-specific dependence).
        Either result composes further with the same operators.

    \item \textbf{Reuse.} The characterized signature (or composite) is reusable across any composition whose contexts fall within the characterized sets.
\end{enumerate}

\subsection{Axioms and Notation}
\label{sec:axiom-summary}

\begin{table}[t]
\centering
\caption{The seven axioms of energy calculus: physical (P) and methodological (M).}
\label{tab:axioms}
\newcolumntype{R}[1]{>{\raggedright\arraybackslash}p{#1}}
\begin{tabular}{ll R{0.33\textwidth} R{0.47\textwidth}}
\toprule
\textbf{Label} & \textbf{No.} & \textbf{Name} & \textbf{Statement (short form)} \\
\midrule
P1 & \ref{ax:P1} & Energy Decomposition & $E_e = s_e T_e + d_e$; all components non-negative. \\
P2 & \ref{ax:P2} & Per-Device Static Power & A device has a single static power; same-device concurrent execution draws it once over the makespan. \\
P3 & \ref{ax:P3} & Interaction Structure & Interaction effects decompose into a finite set of characterizable physical mechanisms. \\
P4 & \ref{ax:P4} & Interaction Locality & Entry-context effects decay with separation between elements. \\
P5 & \ref{ax:P5} & Interference Thresholding & Runtime-context effects are negligible below shared-resource capacity and grow abruptly past saturation. \\
\midrule
M1 & \ref{ax:M1} & Intrinsic Uncertainty & Error is summarized by mean and variance; for elements measured directly, static power is exact and component errors uncorrelated; uncertainty propagates through every composition. \\
M2 & \ref{ax:M2} & Interaction Measurability & Signatures are characterized by measurement under multiple contexts, with other variables fixed. \\
\bottomrule
\end{tabular}
\end{table}

\begin{table}[t]
\centering
\caption{Core notation.}
\label{tab:notation}
\small
\begin{tabular}{@{}ll@{\hspace{2.5em}}ll@{}}
\toprule
\textbf{Symbol} & \textbf{Meaning} & \textbf{Symbol} & \textbf{Meaning} \\
\midrule
$e$ & Energy element & $T_e$ & Execution time (s) \\
$\theta$ & Workload parameters & $d_e$ & Dynamic energy (J) \\
$\kappa$ & Implementation realizing the element's operation & $E_e$ & Total energy $s_e T_e + d_e$ (J) \\
$\omega$ & Operating point & $S_e$ & Energy signature \\
$\omega_{\text{dev}}$ & Device-wide part of $\omega$ & $\sigma^2_{T_e}, \sigma^2_{E_e}$ & Variances of measured time and total energy \\
$\omega_{\text{elem}}$ & Per-element part of $\omega$ & $\sigma^2_{d_e}$ & Variance of dynamic energy \\
$\phi$ & Entry context (state from prior execution) & $\seqop$ & Sequential composition \\
$\Phi$ & Set of characterized entry contexts & $\parop$ & Same-device parallel composition \\
$\phi_{\text{null}}$ & Null context (no state from prior execution) & $\parml$ & Cross-device parallel composition \\
$\rho$ & Runtime context & $c_\alpha$ & Two-sided Gaussian multiplier at confidence $1 - \alpha$ \\
$\mathcal{P}$ & Set of characterized runtime contexts & $\mathcal{D}_e$ & Set of devices hosting the element \\
$s_e$ & Static power of the device hosting $e$ (W) & & \\
\bottomrule
\end{tabular}
\end{table}

This concludes the development of foundational axioms and notation.
\Cref{tab:axioms} summarizes them, and \Cref{tab:notation} collects the core notation used in the remainder of the paper.

%% file: sections/algebra.tex
\section{Properties of Energy Calculus}
\label{sec:algebra}

\Cref{sec:calculus} developed the general context-dependent algebra, and the Reduction Theorem (\Cref{prop:reduction}) established when it reduces to the simpler context-independent algebra.
This section derives useful properties and rules for both.

\subsection{Algebraic Properties}
\label{sec:algebraic-properties}

In this section, we establish the algebraic properties of the composition operators.
All three operators are associative and closed, the parallel operators commute, and sequential composition does not distribute over parallel; these hold in the general context-dependent algebra and therefore in the context-independent algebra as well.\footnote{A context-independent evaluation is a context-dependent evaluation with every signature value taken at the null context in isolation, so identities that hold over all contexts instantiate to the context-independent algebra. Non-distributivity, a negative result, transfers through its counterexample, which is built from null-context point values.}
Additionally, in the context-independent algebra, sequential composition commutes and monotonicity holds.
The properties and their corollaries are summarized in \Cref{tab:properties}.

\begin{table}[t]
\centering
\caption{Properties and corollaries of the composition algebra (short form).
Full statements appear in the body; proofs in appendices.}
\label{tab:properties}
\newcolumntype{R}[1]{>{\raggedright\arraybackslash}p{#1}}
\begin{tabular}{l R{0.17\textwidth} R{0.63\textwidth}}
\toprule
\textbf{No.} & \textbf{Name} & \textbf{Statement (short form)} \\
\midrule
\Cref{prop:commutativity} & Commutativity & $\parop$ and $\parml$ commute unconditionally; $\seqop$ does not commute in general.
\emph{Under the Reduction Theorem}, $\seqop$ commutes in the context-independent algebra (\Cref{cor:ci-commutativity}). \\
\Cref{prop:associativity} & Associativity & All three operators are associative; groupings are equivalent, orderings are not. \\
\Cref{prop:closure} & Closure & Every composition is itself an energy element with a well-defined signature. \\
\Cref{prop:non-distributivity} & Non-Distributivity & $\seqop$ does not distribute over parallel composition; the sides differ by the duplicated dynamic energy. \\
\Cref{prop:monotonicity} & Monotonicity & Component-wise no-smaller signatures can decrease a composition under context dependence.
\emph{Under the Reduction Theorem}, they never decrease any composition in the context-independent algebra (\Cref{cor:ci-monotonicity}). \\
\bottomrule
\end{tabular}
\end{table}

\subsubsection{Commutativity}

Consider a high-power element that leaves the device hot, followed by a thermally sensitive element that throttles when hot (\Cref{sec:mechanisms}).
Also consider the same two elements executed in the opposite order---if the thermally sensitive element runs first on a cooler device, it may avoid throttling and consume different energy.
Therefore, swapping the order of sequential composition can in general change the total energy, unlike parallel compositions that are structurally symmetric.

\begin{property}[Commutativity]
\label{prop:commutativity}
Parallel operators $\parop$ and $\parml$ are unconditionally commutative, in both the context-dependent and context-independent algebras:
\begin{align}
    E_1 \parop E_2 &= E_2 \parop E_1 \\
    E_1 \parml E_2 &= E_2 \parml E_1
\end{align}
Sequential composition does not commute in general:
\begin{align}
    E_1 \seqop E_2 &\neq E_2 \seqop E_1.
\end{align}
\end{property}

\begin{proof}[Proof sketch]
Parallel commutativity follows from symmetry of $\max$ and $+$: both elements evaluate at the shared entry context $\phi_0$ and the shared runtime context $\rho = \{e_1, e_2\}$, leaving no label-dependent term.
For the sequential case, let $\phi_0$ be the shared entry context, with the composition running alone so that runtime contexts are singletons and suppressed.
Then $E_1 \seqop E_2 = s_1(\phi_0) T_1(\phi_0) + d_1(\phi_0) + s_2(\phi_1) T_2(\phi_1) + d_2(\phi_1)$ with $\phi_1$ the exit context of $e_1$'s execution from $\phi_0$, and $E_2 \seqop E_1 = s_2(\phi_0) T_2(\phi_0) + d_2(\phi_0) + s_1(\phi_2) T_1(\phi_2) + d_1(\phi_2)$ with $\phi_2$ the exit context of $e_2$'s execution from $\phi_0$.
In general $\phi_1 \neq \phi_2$, and the two sums differ.
The full proof appears in \Cref{app:proof-commutativity}.
\end{proof}

Non-commutativity reflects context propagation: execution order changes what each element inherits from its predecessor.
When the elements are mutually context-insensitive, however, the Reduction Theorem removes the order dependence.

\begin{corollary}
\label{cor:ci-commutativity}
When $e_1$ and $e_2$ are mutually context-insensitive under both orderings, the Reduction Theorem (\Cref{prop:reduction}) reduces both orderings to the same context-independent value: sequential composition commutes in the context-independent algebra.
The physical orderings themselves agree within their summed reduction bounds, $c_\alpha (\sigma_{E_{e_1}} + \sigma_{E_{e_2}})$ when the composition starts at the null context.
\end{corollary}

Note that the corollary's equality is for total-energy values, not for full composite signatures: two orderings can consume the same energy while leaving different exit contexts, so one ordering substitutes for the other inside a larger composition only when their exit contexts are equivalent over the downstream admissible context set.
Reordering is also available only where the schedule's dependencies permit both orders.

\subsubsection{Associativity}

Whether a three-stage pipeline is modeled as the first two stages followed by the third, or the first followed by the latter two, is bookkeeping: the stages execute in the same order either way, and the total energy is the same.

\begin{property}[Associativity]
\label{prop:associativity}
All three operators are associative, in both context-dependent and context-independent algebras:
\begin{equation}
\begin{aligned}
    (E_1 \seqop E_2) \seqop E_3 &= E_1 \seqop (E_2 \seqop E_3) \\
    (E_1 \parop E_2) \parop E_3 &= E_1 \parop (E_2 \parop E_3) \\
    (E_1 \parml E_2) \parml E_3 &= E_1 \parml (E_2 \parml E_3)
\end{aligned}
\end{equation}
Associativity of~$\seqop$ preserves the order of execution: it is the equivalence of groupings of one composition, not of orderings.
\end{property}

\begin{proof}[Proof sketch]
Associativity of $\parop$ and $\parml$ follows because grouping changes no element's contexts---the same elements co-start and co-run either way---and their energy formulas (\Cref{eq:same-device-parallel,eq:cross-device-parallel}) are built from addition and maximum, both associative.
For $\seqop$, context propagates forward with execution order regardless of grouping, so each element inherits the same context in either grouping.
\Cref{app:proof-associativity} provides the full proof.
\end{proof}

\subsubsection{Closure}

Composition operators take signatures as operands, and they produce a well-defined signature as output.

\begin{property}[Closure]
\label{prop:closure}
A composition of energy elements via any operator is itself an energy element with a well-defined energy signature, in both the context-dependent and context-independent algebras.
\end{property}

\begin{proof}[Proof sketch]
Most signature entries are trivially well-defined: the composite inherits its workload parameters, implementation, and operating point from the operands, and its duration $T'$, dynamic energy $d'$, and duration variance follow the operator's time and dynamic accounting.
The proof shows that the remaining three---static power, energy variance, and the composite's contexts---are also well-defined.
\axname{P1} fixes static power to the non-negative summary $s' = (E_1 \,\mathrm{op}\, E_2 - d')/T'$, a derived quantity rather than a device's baseline draw.
The variance propagated by \Cref{prop:var-propagation} is carried in the signature directly.
The composite's entry and exit contexts follow from its parts, and each operand's runtime context is recovered from the composite's schedule and the operand's position in it.
\Cref{app:proof-closure} writes out the precise formulas and completes the proof.
\end{proof}

Closure enables hierarchical composition.
An element composed within the calculus keeps its term: an enclosing composition still sees its parts, and evaluation derives each constituent's contexts, as the proof sketch describes.
An element measured directly has no parts to see; its signature stands on its own over the contexts it was characterized under ($\Phi$ and $\mathcal{P}$).
Either way, the same operators apply at every level of the hierarchy, and context propagation is handled uniformly.
Within a single composition, some elements may be composed and others directly measured, some fine-grained and others coarse.

\subsubsection{Distributivity}

Applying distributivity to a sequential composition ($\seqop$) over a parallel composition ($\parop$ or $\parml$) would \emph{duplicate} the first element, which is not equivalent to what the original composition executes.

\begin{property}[Non-Distributivity]
\label{prop:non-distributivity}
Sequential composition does \emph{not} distribute over parallel composition, in both the context-dependent and context-independent algebras:
\begin{align}
    E_1 \seqop (E_2 \parop E_3) &\neq (E_1 \seqop E_2) \parop (E_1 \seqop E_3) \\
    E_1 \seqop (E_2 \parml E_3) &\neq (E_1 \seqop E_2) \parml (E_1 \seqop E_3)
\end{align}
\end{property}

\begin{proof}[Proof sketch]
On the right-hand side, both branches contain $e_1$, so $e_1$ executes twice.
Expanding both sides in the context-independent algebra, the static terms coincide---the makespan is the same under either structure---and the two sides differ by exactly $d_1$, the dynamic energy of the duplicated execution.
The full proof appears in \Cref{app:proof-distributivity}.
\end{proof}

\subsubsection{Monotonicity}

When two elements run back to back, the first element's exit context can make the second element cheaper or more expensive than it would be in isolation.
This creates cases that break monotonicity: replacing the earlier element with one that consumes more energy but leaves a context that makes the later element much cheaper (e.g., more cache-friendly output or a cooler device) can lower the total energy of the composition.

\begin{property}[Non-Monotonicity]
\label{prop:monotonicity}
In the context-dependent algebra, total energy is not monotone in a composition's sub-elements: it is possible that replacing a sub-element with one that is no cheaper in any signature component---static power, execution time, and dynamic energy---can \emph{decrease} the composition's total energy.
\end{property}

The failure comes from replacement changing the exit context the rest of the composition inherits: replacing $e_1$ with a more expensive $e'_1$ lowers the total of $e'_1 \seqop e_2$ below that of $e_1 \seqop e_2$ whenever $e'_1$ leaves a context that makes $e_2$ sufficiently cheaper.
\Cref{app:proof-monotonicity} constructs an explicit instance.

In the context-independent algebra, however, every signature value is evaluated at the null context, so a replacement cannot change what the rest of the composition evaluates at.

\begin{corollary}
\label{cor:ci-monotonicity}
In the context-independent algebra, replacing a sub-element with one whose signature has no smaller static power, execution time, and dynamic energy never decreases the energy of any composition that contains it.
Here, replacing means swapping that element's signature alone, with the schedule still feasible, shared operating parameters fixed, and every other signature unchanged.
\end{corollary}

\begin{proof}[Proof sketch]
Each operator's energy formula is monotone in $(s_i, T_i, d_i)$, since addition and maximum are both monotone, and monotonicity carries over to hierarchies because an enclosing composition consumes a sub-composite only through its duration, dynamic energy, and total energy (\Cref{prop:closure}), all weakly increased by such a replacement.
The full proof appears in \Cref{app:proof-monotonicity}.
\end{proof}

\Cref{cor:ci-monotonicity} makes conservative bounds useful: replacing uncertain signatures with upper bounds yields a valid upper bound on the composed energy in the context-independent algebra.

\subsection{Uncertainty and Measurement}
\label{sec:uncertainty}

Every measured signature component carries noise, and a composed prediction inherits the noise of every element it combines.
Physical measurement of energy is authoritative, and divergence between composed predictions and direct measurement signals a composition assumption violation, an incomplete characterization, or a measurement limitation.
This subsection first derives how uncertain a composed prediction is, and then lays out three places where the propagated uncertainty is useful.

\subsubsection{Propagating Uncertainty}

A composed prediction's variance combines the variances of its elements.
This needs the elements' distributions to stay fixed and their measurement errors to be independent.
The first is the context-independent algebra's standing assumption: under mutual context-insensitivity, every execution of an element is a draw from its null-context characterization.\footnote{Each execution of an element is a draw from a distribution, and in the general context-dependent case, the distribution itself shifts with the element's entry and runtime contexts (i.e., not fixed).
However, under mutual context-insensitivity, the distribution is characterized at the null context in isolation and fixed, so every execution is, within the bias tolerance, a draw from that fixed distribution.}
The second is an assumption we add; it holds largely in practice because run-to-run time jitter and per-run dynamic-energy noise belong to each element, and characterization protocols keep measurement conditions (e.g., ambient temperature) fixed across runs (\axname{M2}).

\begin{property}[Variance Propagation]
\label{prop:var-propagation}
The total energy variance of a single directly measured element is given by \Cref{eq:energy-variance}.
In the context-independent algebra, with measurement errors independent between the two elements and the times $(T_1, T_2)$ jointly independent of the dynamic energies $(d_1, d_2)$:
\begin{align}
\label{eq:var-sequential}
    \sigma^2_{E_1 \seqop E_2} &= \sigma^2_{E_1} + \sigma^2_{E_2} \\
    \sigma^2_{E_1 \parop E_2} &= \sigma^2_{d_1} + \sigma^2_{d_2} + s^2\, \mathrm{Var}\bigl(\max(T_1, T_2)\bigr) \\
\label{eq:var-parallel-max}
    \sigma^2_{E_1 \parml E_2} &= \sigma^2_{d_1} + \sigma^2_{d_2} + (s_1 + s_2)^2\, \mathrm{Var}\bigl(\max(T_1, T_2)\bigr)
\end{align}
For $\parop$, $s$ is the shared static power of the device (i.e., $s = s_1 = s_2$).
\end{property}

The sequential rule adds the operands' carried energy variances and takes any operands.
The parallel rules split each operand into its time and dynamic variances, because static power converts time noise into energy noise: when a duration fluctuates by $\Delta T$, energy fluctuates by $\Delta T$ times the static power that keeps flowing through those seconds, and a scale factor on a fluctuation enters its variance squared.
For $\parml$ the weight is $(s_1 + s_2)^2$, because a makespan fluctuation holds both devices longer and their combined static draw flows through it.
This split is why parallel operands must each draw one constant static power, as elements measured directly and constant-baseline composites do (\Cref{prop:closure}).
\Cref{app:variance-sequential,app:variance-parallel} derive the three equations, and \Cref{app:variance-max} gives closed forms and limit cases for $\mathrm{Var}(\max(T_1, T_2))$.

Under the context-dependent algebra, the fluctuation of elements can couple through shared contexts (e.g., a run where $e_1$ runs hot or long shifts the context $e_2$ inherits), and independence fails.
In this case, a contiguous segment containing the coupled elements should be characterized as a single composite (\Cref{sec:interaction-graph}).
The composite's variances are measured directly, independence holds at its boundary, and propagation via \Cref{prop:var-propagation} proceeds between insensitive units as before.

\subsubsection{Using Propagated Uncertainty}

The machinery below consumes a composed prediction $\hat{E}$ and its propagated standard deviation $\sigma_{\hat{E}}$ from \Cref{prop:var-propagation}, and applies under either algebra.

\paragraph{Prediction intervals and comparison.}
From the propagated variance, we can more reliably compare the energy consumption of different configurations based on intervals rather than points.
A composed prediction carries the interval $\hat{E} \pm c_\alpha \cdot \sigma_{\hat{E}}$, covering a single run of the composition at confidence $(1 - \alpha)$ (if it's the average over $n$ repeated runs tightens to $\hat{E} \pm c_\alpha \sigma_{\hat{E}} / \sqrt{n}$).\footnote{Note that the Gaussian shape behind $c_\alpha$ is still an approximation: sums of many weakly dependent errors approach Gaussian by the central limit theorem, but the maxima of parallel composition do not.}
We can confidently say that one configuration consumes less energy than another only when its interval's upper bound lies below the other's lower bound; comparing point predictions alone would let run-to-run noise decide.

\paragraph{Deviation detection.}
Propagated uncertainty also serves as a diagnostic.
We compose signatures to predict a composition's energy, then measure the composition directly.
The two disagree when
\begin{equation}
\label{eq:deviation}
    |\Epred - \Emeas| \;>\; c_\alpha \cdot \sqrt{\sigma^2_{\Epred} + \sigma^2_{\Emeas}},
\end{equation}
where $\sigma^2_{\Epred}$ is the prediction's propagated variance and $\sigma^2_{\Emeas}$ the direct measurement's own variance.
The variance of the difference is the sum of the two variances because the two are independent: the measurement is a fresh run, while the prediction inherits its error from characterization.
Agreement within propagated uncertainty means the composition model is consistent with measurement at the current precision.
Disagreement means an assumption failed: an interaction the composition does not capture, a measurement that needs recalibration, or a signature applied outside its characterized range.
Even under a correct model, the test triggers with probability $\alpha$, and testing many pairs multiplies such false triggers; recharacterization should follow persistent triggers, not single ones.
The workflow of \Cref{sec:interaction-graph} then localizes and characterizes the cause.

The right-hand side of \Cref{eq:deviation} carries random error only.
On top of this random error, the bias accepted by the Reduction Theorem accumulates across composed elements ($\sum_i m_i \, c_\alpha \, \sigma_{E_{e_i}}$, \Cref{app:reduction-accumulation}), and over deep compositions can exceed the propagated interval.
This composite-level check is what catches it.

\paragraph{Refinement and coarsening.}
The deviation test (\Cref{eq:deviation}) also governs granularity changes.
We \emph{refine} an element by decomposing it into sub-elements and composing their signatures, and \emph{coarsen} a set of composed elements into a single element with an aggregate signature; the two are inverses up to measurement error.

\begin{property}[Refinement--Coarsening Consistency]
\label{prop:refinement}
Let element~$e$ have directly measured signature~$S_e$, and let~$e$ be refined into sub-elements $e_1, \ldots, e_n$ composed via the appropriate operators, yielding composed signature~$\hat{S}_e$.
If the sub-elements are mutually context-insensitive, then at confidence level $(1 - \alpha)$,
\begin{equation}
\label{eq:refinement}
    |E_e - \hat{E}_e| \;\leq\; \sum_i m_i \, c_\alpha \, \sigma_{E_{e_i}} \;+\; c_\alpha \cdot \sqrt{\sigma^2_{E_e} + \sigma^2_{\hat{E}_e}},
\end{equation}
where the first term is the composition's accumulated bias bound (\Cref{eq:reduction}); it drops to zero when each $e_i$'s signature is instead correctly characterized at every entry and runtime context the refinement realizes.
\end{property}

When the hypothesis fails, the deviation test triggers, meaning the diagnostic is working as intended.
This gives consistency checks across granularities; the full statement and proof appear in \Cref{app:refinement}.

%% file: sections/frontier.tex
\section{Application: Operating Frontiers Under Context-Independence}
\label{sec:operating-frontier}

We extend energy calculus from composing energy totals to composing time--energy Pareto frontiers.
This section operates under the context-independent algebra, where composed elements are mutually context-insensitive at every implementation and operating point.
Elements that are not mutually context-insensitive are instead characterized together as a single composite (\Cref{sec:interaction-graph}), which this section treats as one element.

\subsection{Operating Set and Frontier}
\label{sec:frontier-defined}

The \emph{operating set}~$\Omega_e$ of element~$e$ is the set of operating points at which $e$ can execute.
An element may also have multiple implementations~$\kappa$ of its operation (\Cref{sec:energy-elements}).
At a fixed workload parameter (thus suppressed from the argument), sweeping the pairs $(\kappa, \omega)$ over the element's implementations and operating points maps each pair to a point on the time--energy plane: $(T_e(\kappa, \omega),\, E_e(\kappa, \omega))$ where $E_e(\kappa, \omega) = s_e(\omega_{\text{dev}})\, T_e(\kappa, \omega) + d_e(\kappa, \omega)$ (\axname{P1}).\footnote{Static power moves with the device-wide operating point $\omega_{\mathrm{dev}}$. DVFS, for instance, changes a device's ready-state static power draw.}
The \emph{achievable set} of $e$ is the set of these swept points, and from the achievable set, we can extract the \emph{operating frontier} $\mathcal{F}_e$ of $e$ by removing Pareto-dominated points.
Frontiers of this form arise at every granularity, from a single GPU kernel to model layers~\cite{kareus:osdi26}, training iterations~\cite{perseus:sosp24}, and entire training jobs~\cite{zeus:nsdi23}.

Composing per-element operating-point schedules assumes that switching between operating points takes negligible time and energy compared to the elements' time and energy.
When switching time is non-negligible, that influences the granularity at which operating-point switching can be applied.
Non-trivial switching time and energy can be accounted for by a transition element $e_{\omega \to \omega'}$, with its own signature composing like any other element.

\subsection{Frontier Composition}
\label{sec:frontier-composition-insensitive}

With each element carrying its own operating set and frontier, we now discuss how the frontiers of composed elements are obtained.
We begin with the three composition operators of \Cref{sec:composition-operators}, and then discuss general composition.

\subsubsection{Sequential Composition}

Consider two elements that run back to back on one device, each free to pick its own implementation and operating point.
Durations add and energies add, so each choice of implementations and operating points lands at the component-wise sum of two swept points.
The set of all pairwise sums of two point sets is therefore their \emph{Minkowski sum}, written:
\begin{equation}
  \mathcal{F}_{e_1} \boxplus \mathcal{F}_{e_2} = \{(T_1 + T_2,\, E_1 + E_2) : (T_1, E_1) \in \mathcal{F}_{e_1},\, (T_2, E_2) \in \mathcal{F}_{e_2}\}.
\end{equation}

\begin{property}[Sequential Frontier Composition]
\label{prop:frontier-sequential}
Let $e_1$ and $e_2$ be mutually context-insensitive at every choice of implementations and operating points.
Under $\seqop$,
\begin{equation}
    \mathcal{F}_{e_1 \seqop e_2} \;=\; \mathrm{Pareto\text{-}min}\bigl(\mathcal{F}_{e_1} \boxplus \mathcal{F}_{e_2}\bigr).
\end{equation}
\end{property}

\begin{proof}[Proof sketch]
Each choice of implementations and operating points sums the two elements' durations and energies, so the achievable set is exactly the pairwise-sum set.
Replacing either element's swept point with a frontier point that dominates it improves the composite in at least one coordinate without worsening the other, so restricting the sweep to the elements' own frontiers loses no composite-optimal point.
The full proof appears in \Cref{app:frontier-sequential}.
\end{proof}

\subsubsection{Parallel Composition}

Now let the two elements run concurrently, on one device or on two.
Whichever element finishes first, every participating device is held until the later one finishes (\Cref{sec:composition-operators}), so static power is drawn for the shared makespan no matter how fast either element runs.
Under $\parop$, a sweep chooses implementations and operating points, with the operating points sharing the device-wide part $\omega_{\text{dev}}$.
Across devices, $\parml$ itself requires equal durations, so wherever a sweep leaves them unequal, an idle element pads the earlier-finishing device to the makespan before $\parml$ applies (\Cref{sec:cross-device-parallel}); each device then draws its own static power, with no shared-$\omega_{\text{dev}}$ constraint.
At each pair of implementations, each operating-point pair evaluates to
\begin{equation}
\begin{aligned}
    E_{e_1 \parop e_2}(\omega_1, \omega_2) &= s(\omega_{\text{dev}}) \cdot T_{\max} + d_1(\omega_1) + d_2(\omega_2) \\
    E_{e_1 \parml e_2}(\omega_1, \omega_2) &= \bigl(s_1(\omega_1) + s_2(\omega_2)\bigr) \cdot T_{\max} + d_1(\omega_1) + d_2(\omega_2)
\end{aligned}
\qquad T_{\max} = \max\bigl(T_1(\omega_1), T_2(\omega_2)\bigr).
\end{equation}
For $\parop$, $s$ is the shared static power of the device ($s = s_1 = s_2$).
For $\parml$, the display takes each padded device's idle draw equal to its active draw; \Cref{app:frontier-parallel} keeps the two distinct.
The makespan term couples the two sweeps, so unlike the sequential case, the composite frontier is not a Minkowski sum.
Instead, it is obtained by sweeping the elements' implementations and operating points in full and keeping the Pareto-optimal points.

The makespan term also gives the composite frontier its right end.
Slowing the longer-running element typically saves its dynamic energy but stretches the makespan, over which every participating device keeps drawing static power.
Once further slowdown saves less dynamic energy than the added static energy over the extra time, total energy turns back upward, and the slower points are Pareto-dominated.
The composite frontier therefore terminates at the energy-minimizing makespan.
Empirically, Perseus~\cite{perseus:sosp24} observes this right-end upturn in its measured iteration-level time--energy curve.

One interesting observation is that when sweeping the two elements' implementations and operating points to obtain points on the composite's time--energy plane, restricting the sweep to the elements' own Pareto-optimal points is unsound: a swept point that is Pareto-suboptimal for an element may be part of a Pareto-optimal point for the composite.
When one element finishes before the other, the window from its finish to the makespan is its \emph{slack}.
Suppose that the earlier-finishing element's device draws the same static power at every swept operating point.
That static power is then drawn for the full makespan no matter what the element does, so the element affects the composite total only through its dynamic energy, not total energy.
However, the element's own frontier is about Pareto-optimal time and \emph{total} energy, not dynamic energy.
The result is that swept points \emph{past} the element's minimum total energy point can be Pareto-optimal in terms of time and \emph{dynamic} energy, and therefore part of the composite's Pareto-optimal frontier.
This generalizes to an operational rule:

\begin{property}[Slack-to-Energy Conversion]
\label{prop:slack-conversion}
Let $e_1$ and $e_2$ be parallel composed elements with durations $T_1 < T_2$, so that the makespan is $T_{\max} = T_2$ and $e_1$ has slack.
Suppose $e_1$ can be slowed, moving its signature components from $(s_1, T_1, d_1)$ to $(s_1', T_1', d_1')$ with $T_1 \leq T_1' \leq T_{\max}$.

\noindent On a shared device, suppose further that the move changes only $e_1$'s implementation or the per-element part of its operating point, leaving the device-wide part, and with it the device's static power and $e_2$'s signature, unchanged.
The move then does not increase composite energy at unchanged makespan whenever:
\begin{enumerate}[nosep, label=(\roman*)]
    \item $d_1' \leq d_1$.
\end{enumerate}
On separate devices, the move does not increase composite energy at unchanged makespan whenever the following hold in addition to (i):
\begin{enumerate}[nosep, label=(\roman*), resume, itemsep=2pt]
    \item $s_1' \leq s_1$; and
    \item $s_1'$ does not exceed the draw of the state $e_1$'s device idles in.
\end{enumerate}
\end{property}

\begin{proof}[Proof sketch]
On a shared device, the hypothesis fixes the device's static power and $e_2$'s signature, and the makespan is unchanged since $T_1' \leq T_{\max}$, so the composite's energy changes by $d_1' - d_1$, non-positive under condition~(i).
On separate devices, the energy change decomposes into three terms: the change in dynamic energy, the change in static energy over the element's execution, and the exchange of idle time for computation over the added duration.
Conditions (i)--(iii) make the corresponding terms non-positive.
The full proof appears in \Cref{app:frontier-parallel}.
\end{proof}

Condition (iii) is the race-to-idle decision.
Filling the slack window with slowed computation wins when idling draws more static power than the slowed computation does.
Hardware with a deep low-power idle state instead makes racing to finish and idling the cheaper schedule.

The actual method for slowing an element depends on how elements are composed and how their operating points are defined.
On a shared device, changing $e_1$'s implementation or the per-element part of its operating point can control its time and dynamic energy; the device-wide part couples both elements' time, dynamic energy, and static power, requiring joint consideration.
On separate devices, device-wide operating points can control elements independently (e.g., DVFS on each device).

\subsubsection{General Composition}

The rules above apply one operator at a time, with the composition term fixed.
Deriving the optimal frontier of a DAG-composed element is a significantly harder problem due to a \emph{global} property: changing the implementation or operating point of one element (and therefore its duration) potentially affects every other element's schedule and amount of slack.
This problem, in fact, contains as a special case the discrete time--cost tradeoff problem of project planning, which is strongly NP-hard~\cite{dtct-complexity} and even APX-hard assuming standard conjectures~\cite{dtct-apx-hardness}.
Perseus addresses this by solving a continuous relaxation of the problem exactly~\cite{perseus:sosp24}.

%% file: sections/discussion.tex
\section{Discussion}
\label{sec:discussion}

This section examines what the framework covers, what lies outside its scope, and what remains open.

\subsection{What the Framework Covers}
\label{sec:coverage}

Energy calculus applies to any computational system that meets four conditions: 
(i) elements admit a static--dynamic decomposition whose first-order error is acceptable at the chosen granularity (\axname{P1}); 
(ii) interactions belong to a finite mechanism inventory, each summarizable by a small number of effective state variables (\axname{P3});
(iii) entry-context effects decay with separation (\axname{P4}) and interference stays below shared-resource capacity (\axname{P5}), so that the interaction graph can remain sparse;
and (iv) measurement instrumentation resolves the chosen granularity, so that the insensitivity tests are informative rather than vacuous (\axname{M1}, \axname{M2}).
Conditions (i) through (iii) ask that the axioms' idealizations hold to acceptable accuracy, and (iv) asks that measurement can detect when they do not (\axname{P2} holds by construction, since device boundaries are drawn where static power is drawn once).
We expect most modern ML workloads (training, inference, fine-tuning), data-processing workloads, and scientific simulations on GPU and CPU hardware to satisfy these conditions, with direct verification discussed in \Cref{sec:open-problems}.

Under them, the calculus supports five capabilities, subject to empirical validation in subsequent work.

\begin{enumerate}[nosep]
    \item \emph{Compositional prediction.}
    Characterized signatures compose algebraically, so the calculus can predict total energy and propagated uncertainty for workload variations (schedule changes, parallelism adjustments, element substitution) without re-running the workload.

    \item \emph{Reordering as a first-class energy lever.}
    Conditional commutativity (\Cref{prop:commutativity}) exposes execution order as an energy lever, since two orderings of the same elements can consume different energy when a driving mechanism is active.
    Characterizing the thermal, cache, and power-delivery mechanisms (\Cref{sec:mechanisms}) turns reordering from an informal heuristic into a quantity the calculus can predict and bound.
    Because equal totals do not imply equal exit contexts, substituting one ordering for another inside a larger composition additionally requires the exit contexts to be equivalent over the downstream admissible set (\Cref{prop:commutativity}).

    \item \emph{Cross-hardware transfer.}
    The compositional structure is hardware-independent, since signatures characterized as functions of implementation $\kappa$, operating point $\omega$, and context $\phi$ on a target accelerator compose under the same operators as those on the reference platform.
    Migration reduces to re-characterizing individual elements on the target, over the target's own implementation set; the algebra does not change.

    \item \emph{Anomaly localization.}
    Deviation detection (\Cref{eq:deviation}) is designed to surface the pair whose predicted energy diverges from measurement.
    The sparsity of the interaction graph, which locality (\axname{P4}) and capacity headroom (\axname{P5}) make possible, restricts the search because a workload-level anomaly reduces in principle to a pair-level or prefix-level diagnosis (\Cref{sec:interaction-graph}).

    \item \emph{Principled coarsening and refinement.}
    Refinement--coarsening consistency (\Cref{prop:refinement}) lets a practitioner trade characterization cost against prediction resolution without changing the algebra.
\end{enumerate}

\subsection{What the Framework Does Not Cover}
\label{sec:non-coverage}

\paragraph{Non-time tradeoffs.}
The calculus reasons in the time--energy plane parameterized by $(\theta, \omega)$.
Tradeoffs that replace the time axis with a non-time objective (accuracy, throughput, model quality) are obtained by composing such objective functions with calculus-computed energy.
A curve in $(A, E)$ space, for example, is $\{(A(\theta),\allowbreak E^\star(\theta)) : \theta \in \Theta\}$, with $A(\theta)$ supplied by the application domain and $E^\star(\theta)$ the minimum energy over the element's implementations and operating points, read off the operating frontier (\Cref{sec:operating-frontier}).
The framework supplies the energy axis of any such tradeoff but does not define the other axes. 

\paragraph{Attribution of held-device static power.}
The composition operators attribute every participating device's static power over the full makespan to the composition, with idle elements covering the gaps (\Cref{sec:composition-operators}).
This is an accounting convention rather than a physical necessity.
It is the appropriate accounting for dedicated allocations but debatable for shared clusters, where an idle device may serve other work.
Alternative attribution policies change how consumed energy is divided among compositions; the calculus computes the consumption itself and leaves the division to the deployment context.

\paragraph{Context driven from outside the composition.}
The general calculus assumes that context evolves under the traces of the composed elements, with exit contexts assembling by the mechanisms' update rules from the composition's own execution (\Cref{sec:mechanisms}).
When context is instead driven by execution outside the composition, this assumption fails.
For example, in a virtualized multi-tenant datacenter, co-located tenants share caches, memory bandwidth, and cooling with the composition, yet their elements are invisible to it, perturbing both the entry context an element inherits and the runtime context it executes in.
The mechanisms themselves are in the inventory and the signature's context dependence is characterizable (\axname{M2}), but the drivers of the realized context are neither measurable nor controllable from within the composition (\Cref{def:contexts}), so the calculus cannot say which characterized context applies.

\paragraph{Runtime adaptation.}
Signatures characterized at a fixed operating point and context do not account for runtime changes (closed-loop thermal throttling that interacts with workload dynamics, autotuning that changes kernel variants within a run, or changes to workload parameters of a subset of elements across iterations).
Such elements also lack a fixed trace (\Cref{def:energy-element}), since their actions depend on what they observe at runtime.
Implementation choice fixed ahead of execution, by contrast, is the signature's implementation entry~$\kappa$ (\Cref{sec:energy-elements}); what falls outside the framework is switching driven by runtime observation.
Adaptive regimes require online signature updates; the framework admits this extension but does not specify it.

\paragraph{Uncharacterized elements.}
The framework predicts only for elements with characterized signatures on the target hardware.
We cannot predict a workload whose elements remain unmeasured; extrapolation from nothing is outside scope.

\paragraph{Sub-first-order physical effects.}
The static--dynamic decomposition is a first-order approximation.
Second-order effects (voltage regulator efficiency curves, temperature-dependent leakage beyond the exponential model, current-draw harmonics) are not resolved by the algebra: their run-varying part folds into $\sigma^2$, and their repeatable part is bias.
These effects inflate the elements' variances and hence the bias tolerances the Reduction Theorem carries; when the inflated tolerances exceed what the application requires, we must refine characterization.

\subsection{Open Problems}
\label{sec:open-problems}

\paragraph{Axiom minimality.}
The seven axioms split into physical (\axname{P1}--\axname{P5}) and methodological (\axname{M1}, \axname{M2}) commitments.
Their independence is open: an independence proof would confirm minimality, whereas a derivation of one axiom from the others would simplify the set.
\axname{P2} is a candidate for demotion, as it is satisfiable by construction whenever device boundaries are drawn at power-domain boundaries, which makes it closer to a definitional commitment than a physical one.

\paragraph{Algebraic placement.}
Conditional commutativity and non-distributivity together place energy calculus in a specific algebraic neighborhood (\Cref{sec:related-work}).
Whether that neighborhood corresponds to a known structure or to a new one remains open.

\paragraph{Worst-case energy bounds.}
The current uncertainty treatment carries mean and variance, with conservative operator-level bounds.
Carrying full distributions, using convolution for sums and order statistics for $\max$, would yield worst-case energy bounds for safety-critical or SLA-bound systems and tighten uncertainty propagation through general DAGs.
Although \Cref{app:variance-distributional} sketches the building blocks, the full integration with the algebra remains open.

\paragraph{Power-aware composition.}
Energy calculus already considers power partially, with static power $s$ a primitive and time-averaged dynamic power within an element recoverable as $d/T$.
However, it does not carry the instantaneous power profile within an element or, derivatively, the peak power composed across elements.
Operationally, peak power matters because it triggers thermal events and power caps, and datacenter supply constrains concurrent draw across devices.

A peak-power axis composed under the same three operators is plausible, but each operator carries a different rule.
Under $\seqop$, the elements are disjoint in time on their device, so the device peak is the $\max$ of individual peaks.
Under $\parop$, the two elements' dynamic power draws overlap in time on the same device, so the composite's dynamic peak is bounded by the sum of the individual dynamic peaks on top of the single static draw (\axname{P2}), attained when the peaks align in time, subject to the device's delivery capacity.
Once the capacity is exceeded, the resulting throttling feedback falls outside the framework (runtime adaptation above).
Under $\parml$, each device hosts only its own element and keeps its own peak; what composes is the aggregate draw across devices, whose peak is bounded by the sum of the per-device peaks and attained when they align in time.
The saturation side of this aggregate constraint is already a runtime-context mechanism (\axname{P5}), and what is missing is the composed power profile that determines when it triggers.
Whether these rules compose cleanly into a second algebra alongside the energy one, or whether the device-capacity bound forces a reformulation, is open.

\paragraph{Interaction-graph sparsity at scale.}
The practical reach of the Reduction Theorem leans on the empirical claim that most element pairs are mutually context-insensitive within uncertainty.
Measuring the interaction graph on diverse training or inference workloads would test the claim directly and quantify how much of a real composition composes context-independently.

\paragraph{Predictive context analysis.}
We populate the interaction graph empirically, admitting pairs only after measured deviation.
Predicting context dependencies from workload structure, without first measuring deviation, would cut characterization cost but requires a physical model of each mechanism.

\paragraph{Operating point interpolation.}
Continuous models $s_e(\omega)$ and $d_e(\theta, \kappa, \omega)$ that predict signatures at uncharacterized operating points would cut characterization cost but shift uncertainty from measurement to interpolation.
The same machinery would support cross-hardware transfer when the target device admits only sparse probes.

\paragraph{Energy-aware compilation.}
Embedding energy calculus into ML compilers, so that code generation optimizes composed energy under latency constraints, would close the loop between characterization and deployment.
A type system whose types carry resource budgets (a graded or quantitative type system) would let static analysis track composed energy and uncertainty alongside time complexity.

\paragraph{Reach beyond GPUs.}
TPUs, NPUs, dataflow accelerators, and analog or in-memory compute may have power-state structures unlike those that motivated the static--dynamic split.
Whether energy calculus applies to these domains, and if so how the axioms and mechanisms adapt, remains open.
Non-ML computational domains, including signal processing and scientific computing on FPGAs, raise the same questions.

\paragraph{Hardware design and architecture exploration.}
Architecture exploration evaluates design choices against performance targets for a set of workloads, using kernel-level modeling and simulation to predict whether a proposed design will hit those targets once built.
As energy becomes a primary bottleneck for ML infrastructure, a design is increasingly likely to carry an energy target alongside its performance target.

Energy calculus supplies the composed energy prediction for such a target, since per-kernel signatures characterized on candidate designs compose under the same operators; a design-space sweep then predicts total energy per workload without re-profiling each candidate end to end.
Operating-point domains are one such design dimension, as CPUs already scale voltage and frequency per core or per core group~\cite{intel-sdm-vol3b, amd-ppr-17h}, which makes per-SM or per-SM-group DVFS on GPUs a design potentially worth building.
Energy calculus already provides the generality to explore its gains, since the device-wide versus per-element split of \Cref{sec:energy-elements} is hardware-dependent; finer frequency domains then move $f$ and $V$ from $\omega_{\text{dev}}$ into $\omega_{\text{elem}}$ and compose with the same operators, with transition latency, itself hardware-dependent, bounding the granularity.
Whether the axioms and mechanisms transfer to pre-silicon models, where signatures come from simulation rather than measurement, is open.

%% file: sections/related.tex
\section{Related Work}
\label{sec:related-work}

\subsection{Energy Measurement and Estimation}
\label{sec:rw-measurement}

Hardware instrumentation (NVIDIA NVML, RAPL, vendor APIs) offers direct energy measurement, with known limitations and biases.
Software such as Zeus~\cite{zeus:github}, ML.ENERGY~\cite{mlenergy-benchmark:neuripsdb25}, and other tools based on NVML~\cite{nvml} provide measurements at varying granularities.
These tools share a common gap: they report totals but provide no algebra for reasoning about how parts contribute to the whole.
Energy calculus takes their measurements as inputs (characterized signatures) and adds the compositional machinery they lack.

\subsection{ML Energy Optimization}
\label{sec:rw-optimization}

Zeus~\cite{zeus:nsdi23} tunes batch size and GPU power limit jointly to navigate the time--energy tradeoff during training.
Perseus~\cite{perseus:sosp24} reduces dynamic energy via GPU frequency scaling on off-critical-path computations in pipeline-parallel training.
Kareus~\cite{kareus:osdi26} jointly optimizes SM allocation, kernel launch timing, and GPU frequency for execution-aware energy optimization.
Many efficiency-focused works further reduce static energy by overlapping computation, communication, and memory operations~\cite{nanoflow-osdi25}.
Each of these optimizes a slice of the problem without a unifying algebraic framework.
Energy calculus provides that framework.
Each system's lever maps onto a signature parameter or composition choice (batch size and power limit into~$\theta$ and~$\omega$, frequency and SM allocation into~$\omega$, scheduling and overlap into the composition structure), so the calculus combines their effects in a single algebraic prediction.

\subsection{Performance and Power Modeling}
\label{sec:rw-modeling}

Performance and power modeling supplies the inputs that energy calculus composes.
Roofline models~\cite{williams2009roofline} and operational intensity analysis characterize compute--memory tradeoffs.
Analytical and learned GPU performance models predict execution time from workload parameters~\cite{forecasting-asplos25}.
Power modeling approaches range from counter-based regressions to detailed microarchitectural models~\cite{accelwattch-micro21}.
These predict $T$, $d$, or $s$ in isolation at specific operating points; energy calculus consumes them as signature components and composes them across elements.

\subsection{Compositional Reasoning in Systems}
\label{sec:rw-compositional}

\paragraph{Network calculus.}
Network calculus~\cite{leboudec-thiran} is the closest compositional neighbor.
It uses min-plus algebra to derive deterministic bounds on traffic and delay through networks of service elements, with compact parameterized representations (arrival curves, service curves) that compose.
Energy calculus pursues the analogous goal for a different physical resource, energy rather than delay, and consequently uses a non-distributive algebra whose sequential operator commutes only conditionally (\Cref{prop:commutativity}), rather than min-plus.
Context propagation is the deeper structural difference: energy calculus hands physical state from element to element through exit and entry contexts, which has no counterpart in network calculus.

\paragraph{Tropical/max-plus algebras.}
Max-plus (tropical) algebras~\cite{baccelli1992synchronization} replace $+$ with $\max$ and $\times$ with $+$, which yields formulations of scheduling and critical-path analysis with both operations commutative and addition distributing over maximum.
Energy calculus is distinct on both counts: commutativity of sequential composition is conditional rather than unconditional (\Cref{prop:commutativity}), and sequential composition does not distribute over parallel (\Cref{prop:non-distributivity}).
The static--dynamic decomposition further couples $\max$ (on time in same-device parallel composition) with $+$ (on dynamic energy), so the same parallel operation simultaneously uses both $\max$ and $+$, a structure outside the reach of max-plus, which collapses the two operations into one.

\paragraph{Category theory.}
Monoidal categories~\cite{selinger:graphical} abstract any system with sequential and parallel composition, so energy calculus fits their shape.
The abstraction deliberately leaves open which additional laws hold between the two operators.
For energy, measurement answers that question, as sequential composition does not distribute over parallel (\Cref{prop:non-distributivity}) and commutes only conditionally (\Cref{prop:commutativity}).

%% file: sections/conclusion.tex
\section{Conclusion}
\label{sec:conclusion}

We propose energy calculus, a compositional framework that makes energy a first-class primitive of computational systems.
Energy elements carrying measured signatures compose under three operators and an accompanying algebra, grounded in seven axioms about static and dynamic power.
Two properties distinguish energy calculus from familiar structures, and both follow from physics rather than algebraic choice.
First, sequential composition does not distribute over parallel, because restructuring a schedule duplicates dynamic energy while static energy follows the makespan.
Second, sequential composition commutes only conditionally, because context propagates between elements through thermal, cache, and power-delivery state.
We also provide a Reduction Theorem that recovers a simple context-independent algebra whenever composed elements are mutually context-insensitive.
Practitioners thus pay the cost of context dependence only where the physics demands it.

Energy calculus is borne of years of building energy optimization systems.
Our systems Zeus, Perseus, and Kareus each exploited a piece of the algebra without naming it, whether the smooth response of energy to workload parameters and operating points, static power drawn over a schedule's full makespan, or the static--dynamic decomposition itself.
The calculus codifies these observations, so that what each system rediscovered in its own corner is stated once, for all of them.

Just as complexity theory and information theory grew from the recognition that time and data admit formal structure, we hope energy calculus is a starting point for an equivalent research direction on energy.
This direction spans empirical validation of the axioms, optimization layers built on the algebra, and the design of hardware whose energy behavior is compositional by construction.
We expect that arrow to reverse, so that where systems once taught us the algebra piece by piece, the algebra can now tell us which systems to build.

%% file: sections/appendices.tex
\section{Interaction Mechanism Details}
\label{app:interactions}

This appendix provides detailed characterization protocols for each mechanism in \Cref{tab:interaction-categories}.

\subsection{Thermal Coupling}

\begin{itemize}[nosep]
    \item \textbf{Mechanism:} An element heats the device; elevated temperature increases leakage current (raising static power) and may trigger dynamic thermal throttling (reducing clock frequency, increasing execution time).
    \item \textbf{Timescale:} Thermal time constants of modern GPUs are on the order of seconds, so thermal interactions matter between elements within a few seconds of each other.
    \item \textbf{Characterization:} Measure $E_e$ at multiple device temperatures, set by varying the duration of a prior thermal load, and fit $E_e(\phi)$ as a function of entry temperature.
        Thermally stable profiling, with cooldown intervals between measurements, isolates this effect.
\end{itemize}

\subsection{Cache and Memory State}

\begin{itemize}[nosep]
    \item \textbf{Mechanism:} Preceding elements populate (warm) or pollute (evict) cache lines that subsequent elements use.
        Cache warming decreases memory traffic and dynamic energy (negative interaction effect); pollution increases it (positive).
    \item \textbf{Timescale:} Microseconds; a few intervening memory-intensive operations overwrite cache state.
    \item \textbf{Characterization:} Measure $E_e$ with a cold cache (after a cache-flushing preamble) and with a warm cache (after a representative preceding element runs).
        The difference is the interaction effect.
\end{itemize}

\subsection{Power Delivery}

\begin{itemize}[nosep]
    \item \textbf{Mechanism:} A sudden change in current draw (e.g., transition between elements with very different power profiles) causes voltage droops in the power delivery network, which may trigger brief frequency reductions or unreliable execution that changes timing.
    \item \textbf{Timescale:} Nanoseconds to microseconds for the electrical transient; the effect on execution may persist for tens of microseconds.
    \item \textbf{Characterization:} Requires high-resolution power measurement; typically visible as a brief execution time anomaly at the boundary between elements.
\end{itemize}

\subsection{Resource Contention}

\begin{itemize}[nosep]
    \item \textbf{Mechanism:} Concurrent elements contend for shared resources, and the contention degrades each element's performance and energy.
        On the same device, the contested resources are memory bandwidth, interconnect bandwidth, and shared cache capacity.
        Across devices, the contested resources are shared fabrics (NVLink, PCIe, or network interconnect) and node- or rack-level power budgets.
    \item \textbf{Timescale:} Instantaneous; contention exists only during concurrent execution.
    \item \textbf{Characterization:} Measure $E_e$ in isolation and concurrently with a representative interferer; the difference is the contention effect.
        Same-device contention is characterized through $\parop$; cross-device interconnect contention is characterized through $\parml$ with a co-running peer that exercises the shared fabric.
        Per-element time under co-execution is directly measurable; per-element dynamic energy can only be measured jointly, since device power is shared, so attributing it across the pair requires time-windowed measurement, counter-based apportioning, or regression across varied peers---when attribution is impractical, characterize the pair as a composite.
\end{itemize}

\section{Insensitivity Expansion}
\label{app:expansion}

This appendix derives the left-hand side of the history-insensitivity condition (\Cref{eq:history-insensitive}) from \axname{P1}.
Fix an element~$e$ at workload parameters~$\theta$, implementation~$\kappa$, and operating point~$\omega$, all suppressed below, and let $\phi, \phi'$ be two entry contexts.
Subtracting \axname{P1}'s decomposition at the two contexts and expanding the static-energy product around $\phi'$:
\begin{align}
E_e(\phi) - E_e(\phi')
&= s_e(\phi)\, T_e(\phi) - s_e(\phi')\, T_e(\phi') + d_e(\phi) - d_e(\phi') \nonumber \\
&= \bigl(s_e(\phi') + s_e(\phi) - s_e(\phi')\bigr)\, \bigl(T_e(\phi') + T_e(\phi) - T_e(\phi')\bigr) - s_e(\phi')\, T_e(\phi') + d_e(\phi) - d_e(\phi') \nonumber \\
&= \bigl(s_e(\phi') + \delta s\bigr)\bigl(T_e(\phi') + \delta T\bigr) - s_e(\phi')\, T_e(\phi') + \delta d \nonumber \\
&= \delta s\, T_e(\phi') + \delta T\, s_e(\phi') + \delta s\, \delta T + \delta d,
\end{align}
where $\delta s \coloneqq s_e(\phi) - s_e(\phi')$, $\delta T \coloneqq T_e(\phi) - T_e(\phi')$, and $\delta d \coloneqq d_e(\phi) - d_e(\phi')$ are the signed shifts.
Taking magnitudes, with $\Delta s = |\delta s|$, $\Delta T = |\delta T|$, and $\Delta d = |\delta d|$ as in \Cref{def:history-insensitive}:
\begin{equation}
\bigl| E_e(\phi) - E_e(\phi') \bigr| \;\leq\; \Delta s\, T_e(\phi') + \Delta T\, s_e(\phi') + \Delta s\, \Delta T + \Delta d
\end{equation}
by the triangle inequality, with equality when every shift is non-negative.
The right-hand side is the left-hand side of \Cref{eq:history-insensitive}: the worst-case energy shift the component deviations can produce.
Bounding it by $c_\alpha \cdot \sigma_{E_e}$ therefore bounds the total energy difference, and, every term being non-negative, each component's contribution separately, e.g., $\Delta T\, s_e(\phi') \leq c_\alpha \cdot \sigma_{E_e}$.
The same expansion yields the left-hand side of the interference-insensitivity condition (\Cref{eq:interference-insensitive}): replace $(\phi, \phi')$ with the co-run and isolation runs $(\phi, \rho)$ and $(\phi, \{e\})$, at each entry context $\phi$ in the definition's set, giving the isolation values as the multipliers.

\section{Reduction Theorem}
\label{app:reduction}

This appendix proves the Reduction Theorem (\Cref{prop:reduction}).
When two elements are mutually context-insensitive within measurement uncertainty, compositions agree with their context-independent evaluations within the theorem's bias tolerances.

\subsection{Setup}
\label{app:reduction-setup}

Fix operating point~$\omega$, workload parameters $\theta_1, \theta_2$, and implementations $\kappa_1, \kappa_2$.
Let $\phi_0 = \phi_{\text{null}}$ be the null context at which signatures are characterized (\axname{M2}), and let $\Phi_1, \Phi_2$ be the sets of entry contexts the composition realizes for elements 1 and 2.
\Cref{def:mutual-context-insensitive} includes the null context, so $\phi_0 \in \Phi_i$.
Write $\phi_c$ for the composite's own entry context, which lies in $\Phi_i$ for every element that starts with the composite.
The Reduction Theorem (\Cref{prop:reduction}) assumes mutual context-insensitivity (\Cref{def:mutual-context-insensitive}); its history-insensitivity requirement (\Cref{def:mutual-history-insensitive}) with confidence level~$(1 - \alpha)$ reads:
\begin{equation}
\label{eq:reduction-component-bound}
    \Delta s_i\, T_i(\phi_0) + \Delta T_i\, s_i(\phi_0) + \Delta s_i\, \Delta T_i + \Delta d_i \;\leq\; c_\alpha \cdot \sigma_{E_i} \qquad \forall \phi \in \Phi_i, \; i \in \{1, 2\},
\end{equation}
where $\Delta s_i = |s_i(\phi_0) - s_i(\phi)|$, $\Delta T_i = |T_i(\phi_0) - T_i(\phi)|$, and $\Delta d_i = |d_i(\phi_0) - d_i(\phi)|$, writing $\sigma_{E_i}$ for $\sigma_{E_{e_i}}$.
This bounds the perturbation of $s_i$, $T_i$, and $d_i$ separately, which is what the proofs below require.
It implies the total-energy bound
\begin{equation}
\label{eq:reduction-energy-bound}
    \bigl| E_i(\theta_i, \kappa_i, \omega, \phi) - E_i(\theta_i, \kappa_i, \omega, \phi_0) \bigr| \;\leq\; c_\alpha \cdot \sigma_{E_i} \qquad \forall \phi \in \Phi_i, \; i \in \{1, 2\}
\end{equation}
by the expansion of \Cref{app:expansion}.
When measurement uncertainty is dominated by a single component, the two bounds coincide; otherwise the component-wise bound is the stronger condition that the theorem states.

\subsection{Sequential Case}

\begin{proof}
The context-dependent sequential composition (\Cref{eq:sequential}) from the composite's entry context $\phi_c$ is
\begin{equation}
    E_1 \seqop E_2 = s_1(\phi_c) T_1(\phi_c) + d_1(\phi_c) + s_2(\phi_1) T_2(\phi_1) + d_2(\phi_1)
\end{equation}
with $\phi_1 \in \Phi_2$ the entry context the second element sees, the exit context of $e_1$'s execution from $\phi_c$; the composition runs alone, so runtime contexts are singletons and suppressed.
The context-independent composition evaluates every value at $\phi_0$:
\begin{equation}
    E_1 \seqop^{(0)} E_2 = s_1(\phi_0) T_1(\phi_0) + d_1(\phi_0) + s_2(\phi_0) T_2(\phi_0) + d_2(\phi_0).
\end{equation}
Element 2's block shifts by at most one tolerance,
\begin{equation}
    |s_2(\phi_1) T_2(\phi_1) - s_2(\phi_0) T_2(\phi_0)| + |d_2(\phi_1) - d_2(\phi_0)| \leq c_\alpha \cdot \sigma_{E_2},
\end{equation}
by \Cref{eq:reduction-component-bound} applied to $e_2$ at $\phi_1 \in \Phi_2$, via the expansion of \Cref{app:expansion}.
Element 1's block shifts by at most its own tolerance by the same bound applied at $\phi_c \in \Phi_1$, giving $m_1 = m_2 = 1$; when the composition starts at the null context, $\phi_c = \phi_0$, element 1's block is exact, and $m_1 = 0$.
\end{proof}

\subsection{Parallel Cases}

\begin{proof}
Under~$\parop$ and~$\parml$, both elements start at the composite's entry context $\phi_c$ and share the runtime context, so two substitutions arise: entry-state substitution (evaluating at $\phi_0$ rather than $\phi_c$) and isolation substitution (evaluating each element in isolation rather than at the composition's shared runtime context).
Mutual history-insensitivity (\Cref{eq:reduction-component-bound}) bounds the first and mutual interference-insensitivity (\Cref{def:mutual-interference-insensitive}) the second, each perturbing an element's static power, time, and dynamic energy by at most one tolerance.
Superscript $(0)$ marks null-context isolation values; unmarked values are realized ones.

Under~$\parop$, the composite is $s \cdot \max(T_1, T_2) + d_1 + d_2$, and its shift is at most $\Delta s \cdot \max(T_1^{(0)}, T_2^{(0)}) + s^{(0)} \, |\Delta \max(T_1, T_2)| + \Delta d_1 + \Delta d_2$.
The first term is the longer element's $\Delta s$ summand, the second at most the shifted element's $\Delta T$ summand ($|\Delta \max| \leq \max_i |\Delta T_i|$), so each element contributes at most one tolerance per substitution: $m_1 = m_2 = 2$, dropping to one each when the composition starts at the null context and the entry-state substitution is unneeded.

Under~$\parml$, the durations are equal, the composite is $s_1 T_1 + d_1 + s_2 T_2 + d_2$, and each element's block carries at most one tolerance per substitution exactly as under $\parop$.
\end{proof}

\subsection{Tightness and Breakdown}

The bound is tight when both elements are maximally sensitive within the admitted uncertainty envelope; in practice the deviation is smaller.

\paragraph{Thermal example.}
Static power depends on device temperature~$T_{\text{dev}}$ through leakage current, which rises with temperature.
Across the operating range of typical GPU workloads ($\sim$50--85$^\circ$C for an A100 under sustained load), the dependence is well approximated by the linear form
\begin{equation}
    s(T_{\text{dev}}) \approx s_0 \cdot \bigl(1 + \beta (T_{\text{dev}} - T_{\text{ref}})\bigr),
\end{equation}
where $T_{\text{ref}}$ is the reference temperature at which we measured $s_0$ and $\beta$ is a device-specific leakage coefficient on the order of $0.5$--$1\%$ per~$^\circ$C for current-generation datacenter GPUs.
Elements that run in a narrow thermal band (e.g., within a few~$^\circ$C of $T_{\text{ref}}$) perturb $s$ by less than typical measurement uncertainty across $\Phi_1, \Phi_2$, so the Reduction Theorem applies.
Elements that span a wide thermal range (e.g., a cold-start kernel followed by a sustained compute burst spanning 20--30$^\circ$C, producing a $10$--$30\%$ swing in static power) violate the bound and require $\phi$-parameterized characterization.

\section{Bias Accumulation}
\label{app:reduction-accumulation}

The Reduction Theorem's pairwise bounds (\Cref{eq:reduction}) add over composition, one term per element.
For a shared-device chain $e_1 \seqop \cdots \seqop e_n$ run alone from the null context, with every element history-insensitive over the entry contexts the chain realizes together with the null context, applying the sequential case of \Cref{app:reduction} to each element in turn gives
\begin{equation}
\label{eq:bias-accumulation}
    \bigl| E_{e_1 \seqop \cdots \seqop e_n} - E^{(0)} \bigr| \;\leq\; \sum_{i=2}^{n} c_\alpha\, \sigma_{E_i},
\end{equation}
where $E^{(0)}$ is the context-independent value; element 1 executes at $\phi_0$ and contributes no term.
Parallel operands contribute their operator's multiplied tolerances (\Cref{eq:reduction}).
Multipliers do not compound: composing a composite with a further element adds the new element's multiplied tolerance and leaves the accumulated bias untouched, so the total is a sum over elements, not a product over levels.

The composition $(e_1 \seqop \iota) \parml e_2$, with $\iota$ the idle element restoring equal durations (\Cref{sec:cross-device-parallel}), is the accumulation's canonical instance beyond chains.
Superscript $(0)$ marks null-context isolation values.
Write $\phi_\iota$ for the idle element's entry context, the exit context of $e_1$'s execution; \Cref{eq:sequential,eq:cross-device-parallel} give the two elements' terms plus the idle energy $s_1(\phi_\iota, \{e_2\}) \, (T_2 - T_1)$.
Device 1 is held for the makespan whatever $T_1$ is, so its static energy forms one ledger: with $\delta_a = s_1(\phi_0, \rho) - s_1^{(0)}$ and $\delta_\iota = s_1(\phi_\iota, \{e_2\}) - s_1^{(0)}$ the draw deviations during $e_1$'s execution and during the idle window,
\begin{equation}
    s_1(\phi_0, \rho)\, T_1 + s_1(\phi_\iota, \{e_2\})\, (T_2 - T_1) \;=\; s_1^{(0)}\, T_2 + \delta_a\, T_1 + \delta_\iota\, (T_2 - T_1),
\end{equation}
and $T_1$ survives only against the deviations: a $T_1$ shift moves the boundary between two draws that already agree within tolerance, a second-order effect.
Subtracting the context-independent value $E^{(0)}$ from the composition's energy $E$ and collecting the first-order terms,
\begin{equation}
    \bigl| E - E^{(0)} \bigr| \;\leq\; |\delta_a|\, T_1^{(0)} + |\Delta d_1| + \bigl| s_2 T_2 + d_2 - s_2^{(0)} T_2^{(0)} - d_2^{(0)} \bigr| + s_1^{(0)}\, |\Delta T_2| + |\delta_\iota| \cdot \bigl(T_2^{(0)} - T_1^{(0)}\bigr).
\end{equation}
The first two terms are $e_1$'s $\Delta s$ and $\Delta d$ summands, at most one tolerance per substitution, contributing $2 c_\alpha \sigma_{E_1}$.
Element 2's block is at most $2 c_\alpha \sigma_{E_2}$ by the expansion of \Cref{app:expansion}, and $s_1^{(0)} |\Delta T_2| \leq (s_1^{(0)}/s_2^{(0)}) \cdot 2 c_\alpha \sigma_{E_2}$, since device 1 keeps drawing while $T_2$ stretches.
The idle element's own conditions (\Cref{def:history-insensitive,def:interference-insensitive}) bound the last term: its duration is schedule-set and its dynamic energy zero, so \Cref{eq:history-insensitive,eq:interference-insensitive} reduce to their $\Delta s$ terms and give $|\delta_\iota| \cdot (T_2^{(0)} - T_1^{(0)}) \leq 2 c_\alpha \sigma_{E_\iota}$ across the entry and isolation substitutions, with $\sigma^2_{E_\iota} = s_1^{(0)\,2} \, (\sigma^2_{T_1} + \sigma^2_{T_2})$ inherited from the independent durations that set the gap (\Cref{sec:composition-operators}).
The composition therefore accumulates, to first order, $2\, c_\alpha \sigma_{E_1} + 2\,(1 + s_1/s_2)\, c_\alpha \sigma_{E_2} + 2\, c_\alpha \sigma_{E_\iota}$ at null-context values: the idle element contributes twice its own tolerance, and $e_2$, whose time gates the idle window, carries the idle device's draw as the power ratio.
Products of deviations with time shifts are second order and dropped, as is a realized ordering flip, possible only when the null-context idle window is itself within time noise.

The bias grows linearly in $n$ while propagated uncertainty grows as $\sqrt{n}$: on a shared device with independent errors, $\sigma^2_{E^{(0)}} = \sum_i \sigma^2_{E_i}$ (\Cref{eq:var-sequential}), and by Cauchy--Schwarz $\sum_i \sigma_{E_i} \leq \sqrt{n} \cdot \sqrt{\sum_i \sigma^2_{E_i}}$, with equality at equal variances.
The reported interval therefore does not cover the worst-case accumulated bias at depth; an interval that does must add the accumulated bias bound to the noise quantile,
\begin{equation}
    \bigl| E - E^{(0)} \bigr| \;\leq\; \sum_{i=2}^{n} c_\alpha\, \sigma_{E_i} \;+\; c_\alpha\, \sigma_{E^{(0)}}
\end{equation}
at confidence $(1 - \alpha)$, with $E$ the physical energy of the chain.

Attaining the linear worst case requires per-element deviations of aligned sign.
Deviations that are independent and zero-mean across elements add in quadrature and are absorbed by the propagated interval; aligned signs mean systematic context drift, and sustained drift moves entry contexts out of their characterized sets, so the regime that realizes the worst case is also the regime the applicability condition (\Cref{sec:interaction-graph}) excludes and deviation detection (\Cref{eq:deviation}) flags at composite granularity.

Finally, the idle window's power ratio mirrors the amplification that propagated variance itself carries (the $s_1 + s_2$ weight in \Cref{eq:var-parallel-max}), so measured in multiples of the composite's own propagated deviation, each element's bias contribution stays $O(1)$ regardless of the participating devices' static powers.
Depth, not device configuration, is what loosens the guarantee, with idle windows contributing like elements, at their own tolerances.

\section{Energy Calculus Algebraic Property Proofs}
\label{app:proofs}

We provide full proofs for the algebraic properties stated in \Cref{sec:algebraic-properties}.

\subsection{Commutativity (\Cref{prop:commutativity})}
\label{app:proof-commutativity}

\begin{proof}
We prove the parallel case first, then the sequential case.

\emph{Same-device parallel ($\parop$):}
$E_1 \parop E_2 = s \cdot \max(T_1, T_2) + d_1 + d_2 = s \cdot \max(T_2, T_1) + d_2 + d_1 = E_2 \parop E_1$
by commutativity of $\max$ and addition.
Both elements evaluate at the shared entry context $\phi_0$ and the shared runtime context $\rho = \{e_1, e_2\}$, both symmetric in the two elements, so no label-dependent term arises.
Hence parallel commutativity holds unconditionally.

\emph{Cross-device parallel ($\parml$):}
$E_1 \parml E_2 = s_1 T_1 + d_1 + s_2 T_2 + d_2 = E_2 \parml E_1$
by commutativity of addition, the expression being symmetric in the two elements, with the same shared-$\phi_0$ argument.

\emph{Sequential ($\seqop$):}
Let $\phi_0$ be the entry context of the composite; the composition runs alone, so runtime contexts are singletons and suppressed.
Then
\begin{equation}
    E_1 \seqop E_2 = s_1(\phi_0) T_1(\phi_0) + d_1(\phi_0) + s_2(\phi_1) T_2(\phi_1) + d_2(\phi_1)
\end{equation}
where $\phi_1$ is the exit context of $e_1$'s execution from $\phi_0$, and similarly
\begin{equation}
    E_2 \seqop E_1 = s_2(\phi_0) T_2(\phi_0) + d_2(\phi_0) + s_1(\phi_2) T_1(\phi_2) + d_1(\phi_2)
\end{equation}
with $\phi_2$ the exit context of $e_2$'s execution from $\phi_0$.
In general $\phi_1 \neq \phi_2$, and the two sums differ; the thermal example preceding \Cref{prop:commutativity} provides an explicit instance.
Equality between orderings can also hold outside the context-insensitive case when the context-dependent perturbations of $e_1$ and $e_2$ cancel.
\end{proof}

\begin{proof}[Proof of \Cref{cor:ci-commutativity}]
Assume $e_1$ and $e_2$ are mutually context-insensitive under both orderings (\Cref{def:mutual-context-insensitive}), with the composition starting at the null context ($\phi_0 = \phi_{\text{null}}$): each element's history set contains $\phi_0$, seen when it runs first, and the context the other produces, seen when it runs second.
By \Cref{def:history-insensitive}, substituting $\phi_0$ for $\phi_1$ (respectively $\phi_2$) perturbs each element's total contribution by at most $c_\alpha \cdot \sigma_{E_i}$.
Substituting $\phi_0$ for $\phi_1$ and $\phi_2$ in both expressions reduces both orderings to the same context-independent value,
\begin{equation}
    E_1 \seqop^{(0)} E_2 = s_1(\phi_0) T_1(\phi_0) + s_2(\phi_0) T_2(\phi_0) + d_1(\phi_0) + d_2(\phi_0) = E_2 \seqop^{(0)} E_1,
\end{equation}
so the reduced evaluation is exactly commutative.
The physical orderings are each within one element's tolerance of this common value ($e_2$'s in the order $e_1 \seqop e_2$, $e_1$'s in the other), so $|E_1 \seqop E_2 - E_2 \seqop E_1| \leq c_\alpha (\sigma_{E_1} + \sigma_{E_2})$, establishing the corollary.
\end{proof}

\subsection{Associativity (\Cref{prop:associativity})}
\label{app:proof-associativity}

\begin{proof}
\emph{Sequential ($\seqop$):}
$(E_1 \seqop E_2) \seqop E_3 = \sum_{i=1}^{3}\bigl(s_i(\phi_{i-1}) T_i(\phi_{i-1}) + d_i(\phi_{i-1})\bigr)$
where $\phi_0$ is the composite entry context and $\phi_i$ the exit context of $e_i$'s execution from $\phi_{i-1}$, with the composition running alone so that runtime contexts are singletons and suppressed.
The same sum arises from $E_1 \seqop (E_2 \seqop E_3)$ for two reasons: (i)~context propagation follows execution order, which grouping does not change, so the context entering $e_3$ in the grouping $(12)3$ equals the context entering $e_3$ in $1(23)$; and (ii)~total energy is a sum of per-element contributions that depend only on each element's individual entry context, with no grouping-dependent term.
Both reductions therefore produce the same ordered sequence of context arguments and the same total energy.

\emph{Same-device parallel ($\parop$):}
All three elements co-start and co-run under either grouping, so every value evaluates at the shared entry context $\phi_0$ and the same runtime context $\rho$, and
$(E_1 \parop E_2) \parop E_3 = s \cdot \max(T_1, T_2, T_3) + d_1 + d_2 + d_3 = E_1 \parop (E_2 \parop E_3)$
by associativity of $\max$ and addition.

\emph{Cross-device parallel ($\parml$):}
Likewise with every value at the shared entry context and runtime context, and all durations equal by the operator's requirement,
$(E_1 \parml E_2) \parml E_3 = (s_1 + s_2 + s_3) \cdot T_1 + d_1 + d_2 + d_3 = E_1 \parml (E_2 \parml E_3)$
by associativity of $+$.
\end{proof}

\subsection{Closure (\Cref{prop:closure})}
\label{app:proof-closure}

\begin{proof}
For each operator, we exhibit composite components $(s', T', d')$ that are non-negative, satisfy \axname{P1} ($E' = s' \cdot T' + d'$), and carry the context bookkeeping the composite needs to participate in further compositions.

\emph{Sequential ($\seqop$):}
With every value evaluated as in \Cref{eq:sequential},
\begin{equation}
    T' = T_1(\phi_0, \rho_1) + T_2(\phi_1, \rho_2), \qquad
    d' = d_1(\phi_0, \rho_1) + d_2(\phi_1, \rho_2),
\end{equation}
and \axname{P1} then fixes the static power to
\begin{equation}
    s' = \frac{E' - d'}{T'} = \frac{s_1(\phi_0, \rho_1)\,T_1(\phi_0, \rho_1) + s_2(\phi_1, \rho_2)\,T_2(\phi_1, \rho_2)}{T_1(\phi_0, \rho_1) + T_2(\phi_1, \rho_2)},
\end{equation}
the time-averaged static power, reducing to the device's single $s$ when both elements draw it.
Given $T'$ as the physical makespan and $d'$ as the summed dynamic energy, this is the unique decomposition preserving total static energy $s' \cdot T'$ and total dynamic energy $d'$.
The summary $s'$ is a derived quantity, computed once from the operands' signature values; when $s_1 \neq s_2$, it varies with the operands' durations and hence with their workload parameters, unlike the baseline draw of a device (\axname{P1}).
Every operator consumes an operand's static power only as static energy, so the summary stays exact in further composition: $\seqop$ multiplies $s'$ by the composite's own duration $T'$; $\parml$ requires equal durations, so its makespan window equals $T'$; and $\parop$ requires its operands to run on one device at a shared device-wide operating point, drawing a single $s$.

\emph{Same-device parallel ($\parop$):}
With both elements evaluated at the shared $(\phi_0, \rho)$ of \Cref{eq:same-device-parallel},
\begin{equation}
    s' = s(\phi_0, \rho), \quad
    T' = \max\bigl(T_1(\phi_0, \rho), T_2(\phi_0, \rho)\bigr), \quad
    d' = d_1(\phi_0, \rho) + d_2(\phi_0, \rho),
\end{equation}
where $s$ is the device's static power (\axname{P2}), and $s' \cdot T' + d'$ is \Cref{eq:same-device-parallel} itself.

\emph{Cross-device parallel ($\parml$):}
With both elements evaluated at the shared $(\phi_0, \rho)$ of \Cref{eq:cross-device-parallel} and durations equal by the operator's requirement,
\begin{equation}
\begin{aligned}
    s' &= s_1(\phi_0, \rho) + s_2(\phi_0, \rho), \qquad
    T' = T_1(\phi_0, \rho) = T_2(\phi_0, \rho), \\
    d' &= d_1(\phi_0, \rho) + d_2(\phi_0, \rho),
\end{aligned}
\end{equation}
and $s' \cdot T' + d'$ is \Cref{eq:cross-device-parallel} itself; the composite's device set is the union $\mathcal{D}_{e_1} \cup \mathcal{D}_{e_2}$.

In all cases $s', T', d' \geq 0$, since each combines non-negative values by addition, maximum, and time-averaging.
The composite's entry context is that of its first element (sequential) or of its shared start (parallel), tracked per participating device; its exit context is the state its internal dynamics leave behind; its runtime context tracks the composite's internal timing, with peers entering with the windows they overlap, so each constituent's $\rho_i$ is recovered by interval overlap against the composite's internal boundaries.
The composite's uncertainty is likewise carried directly; its energy variance is the propagated $\sigma^2_{E'}$ of \Cref{prop:var-propagation}, alongside the duration variance the operator's time accounting yields (e.g., $\sigma^2_{T'} = \sigma^2_{T_1} + \sigma^2_{T_2}$ under $\seqop$ with independent errors).
Where that property's independence assumptions fail, the covariance form (\Cref{eq:var-sequential-cov}) or direct measurement of the coupled composite supplies the energy variance (\Cref{sec:uncertainty}).
When the operands draw one shared static power, the composite has a constant baseline and \Cref{eq:energy-variance} applies to it like an element measured directly.
When their static powers differ, it does not: subtracting $s'^2\,\sigma^2_{T'}$ from $\sigma^2_{E'}$ would weight every operand's time noise by the average $s'$, whereas each operand's time noise converts into energy noise through that operand's own static power.
Further composition consumes the carried variances the same way it consumes a measured element's; under $\seqop$ with independent errors, the operands' energy variances add (\Cref{app:variance-sequential}).
The result extends to arbitrary compositions by induction on composition depth, each level composing the previous level's components by the same rules.
\end{proof}

\subsection{Distributivity (\Cref{prop:non-distributivity})}
\label{app:proof-distributivity}

\begin{proof}
\emph{Same device ($\parop$):}
Expand both sides under $E_X \parop E_Y = s \cdot \max(T_X, T_Y) + d_X + d_Y$, with all signatures evaluated at the null context and elements sharing the device's static power $s$.
The left-hand side composes $e_1$ sequentially with the parallel pair $(e_2, e_3)$:
\[
    E_1 \seqop (E_2 \parop E_3) = s(T_1 + \max(T_2, T_3)) + d_1 + d_2 + d_3.
\]
The right-hand side composes $e_1$ sequentially with each of $e_2$ and $e_3$ before the parallel join:
\[
    (E_1 \seqop E_2) \parop (E_1 \seqop E_3) = s \cdot \max(T_1 + T_2, T_1 + T_3) + 2 d_1 + d_2 + d_3.
\]
Since $\max(T_1 + T_2, T_1 + T_3) = T_1 + \max(T_2, T_3)$, the static terms coincide, and the two sides differ by exactly $d_1$, the dynamic energy of the duplicated execution of $e_1$; whenever $d_1 > 0$, the equality fails.

\emph{Cross device ($\parml$):}
Place $e_2$ on device~$\beta$ and $e_3$ on device~$\gamma$, with $T_2 = T_3$ per $\parml$'s equal-duration requirement.
For the left-hand side, one copy of $e_1$ runs on~$\beta$ while $\gamma$ hosts an idle element, so both devices draw static power for the full makespan:
\[
    E_1 \seqop (E_2 \parml E_3) = (s_\beta + s_\gamma)\bigl(T_1 + T_2\bigr) + d_1 + d_2 + d_3.
\]
On the right-hand side both branches contain $e_1$, and $\parml$ places the branches on different devices, so the expression is well-typed only under replication, a device-indexed copy of $e_1$ executing on each device; with the copies sharing $(T_1, d_1)$ on identical devices,
\[
    (E_1 \seqop E_2) \parml (E_1 \seqop E_3) = (s_\beta + s_\gamma)\bigl(T_1 + T_2\bigr) + 2 d_1 + d_2 + d_3,
\]
and the two sides again differ by exactly $d_1$, the dynamic energy of the additional copy.
On heterogeneous devices, the copies' signatures differ, and the gap is the additional copy's own dynamic energy.
\end{proof}

\subsection{Monotonicity (\Cref{prop:monotonicity})}
\label{app:proof-monotonicity}

\begin{proof}
Let $e_1 \seqop e_2$ run alone from entry context $\phi_0$, with $\phi_1$ the exit context of $e_1$'s execution.
Take $e'_1$ matching $e_1$'s static power and execution time but consuming $\delta > 0$ more dynamic energy, whose exit context $\phi'_1$ lowers the successor's energy by more: $E_2(\phi'_1) \leq E_2(\phi_1) - 2\delta$, e.g., through a more cache-friendly output layout.
Replacing $e_1$ with the component-wise no-smaller $e'_1$ then decreases the total: $(E_1 + \delta) + E_2(\phi'_1) \leq E_1 + E_2(\phi_1) - \delta$.
\end{proof}

\begin{proof}[Proof of \Cref{cor:ci-monotonicity}]
In the context-independent form (\Cref{sec:ci-form}), signature values are fixed point values across the substitution, independent of the contexts the composition produces.
Replacing means swapping $e_i$'s signature alone, with the schedule still feasible, shared operating parameters fixed, and every other signature unchanged; a replacement that changes the shared device-wide operating point is treated separately below.
We show that if element~$e'_i$ has $s'_i \geq s_i$, $T'_i \geq T_i$, and $d'_i \geq d_i$, then replacing $e_i$ with $e'_i$ in any composition weakly increases total energy.

\emph{Sequential ($\seqop$):}
$E(e'_1, e_2) = s'_1 T'_1 + s_2 T_2 + d'_1 + d_2 \geq s_1 T_1 + s_2 T_2 + d_1 + d_2$
because $s'_1 T'_1 \geq s_1 T_1$ (product of non-negative numbers, both factors weakly larger) and $d'_1 \geq d_1$.

\emph{Same-device parallel ($\parop$):}
$E(e'_1, e_2) = s \cdot \max(T'_1, T_2) + d'_1 + d_2 \geq s \cdot \max(T_1, T_2) + d_1 + d_2$
because $\max$ is monotone in each argument.
A replacement that raises the shared operating state also raises $s$ for the whole composite, which again weakly increases the total.

\emph{Cross-device parallel ($\parml$):}
$E(e'_1, e_2) = (s'_1 + s_2) \cdot \max(T'_1, T_2) + d'_1 + d_2$, with an idle element at its device's draw restoring equal durations on the shorter side; the static-energy term is weakly larger because it is a product of weakly larger non-negative factors.

The result extends to arbitrary compositions by induction on composition depth: an enclosing composition consumes a sub-composite only through its duration, dynamic energy, and total energy, each of which the replacement weakly increases.
\end{proof}

\section{Variance Propagation Details}
\label{app:variance}

This appendix develops the variance propagation results summarized in \Cref{sec:uncertainty}.
Throughout, operands are elements measured directly or constant-baseline composites, so per \axname{M1}, static powers, including those of idling devices, are exact and each element's time and dynamic-energy errors are uncorrelated.

\subsection{Sequential Composition}
\label{app:variance-sequential}

Under $\seqop$, total energy is the sum of the operands' energies, so the variance of the sum yields
\begin{equation}
\label{eq:var-sequential-cov}
    \sigma^2_{E_1 \seqop E_2} = \sigma^2_{E_1} + \sigma^2_{E_2} + 2\,\mathrm{Cov}\bigl(E_1, E_2\bigr),
\end{equation}
which reduces under independent errors to \Cref{eq:var-sequential}; for operands within \Cref{eq:energy-variance}'s scope, each term expands as $\sigma^2_{E_i} = s_i^2\, \sigma^2_{T_i} + \sigma^2_{d_i}$, with $s_i$ the device's draw at element~$i$'s operating state.

\subsection{Parallel Composition}
\label{app:variance-parallel}

Under $\parop$, $E_{12} = s \cdot \max(T_1, T_2) + d_1 + d_2$.
Writing $\sigma^2_{E_i} = s^2 \mathrm{Var}(T_i) + \mathrm{Var}(d_i)$ for each element and expanding $\mathrm{Var}(E_{12})$ under measurement errors independent between the two elements replaces the per-element $s^2 \mathrm{Var}(T_i)$ terms with $s^2 \mathrm{Var}(\max(T_1, T_2))$:
\begin{equation}
\label{eq:var-parallel}
    \sigma^2_{E_1 \parop E_2} = \sigma^2_{E_1} + \sigma^2_{E_2} + s^2 \cdot \mathrm{Var}\bigl(\max(T_1, T_2)\bigr) - s^2 \cdot \bigl(\mathrm{Var}(T_1) + \mathrm{Var}(T_2)\bigr) + 2s \cdot C_{\max},
\end{equation}
where $C_{\max}$ is the covariance between $\max(T_1, T_2)$ and $d_1 + d_2$; when the times $(T_1, T_2)$ are jointly independent of the dynamic energies $(d_1, d_2)$, $C_{\max} = 0$, and \Cref{eq:var-parallel} reduces to \Cref{prop:var-propagation}'s form.

Under $\parml$, the same expansion gives $\sigma^2_{E_1 \parml E_2} = \sigma^2_{E_1} + \sigma^2_{E_2} + (s_1 + s_2)^2\, \mathrm{Var}(\max(T_1, T_2)) - s_1^2\, \mathrm{Var}(T_1) - s_2^2\, \mathrm{Var}(T_2) + 2\,(s_1 + s_2)\, C_{\max}$: the makespan term carries the combined weight $(s_1 + s_2)^2$, while each subtracted per-element term keeps its own $s_i^2$.
The cross-device form takes each held device to keep drawing its active static power through the makespan.
The exact expression involves the distribution of $\max(T_1, T_2)$, which depends on the marginal distributions of $T_1$ and $T_2$.
For known distributions (e.g., Gaussian execution times), we can compute $\mathrm{Var}(\max(T_1, T_2))$ analytically or via simulation.
When the measurement noise on $E_1$ and $E_2$ is correlated (e.g., shared instrumentation, shared thermal history), an additional $2\,\mathrm{Cov}(E_1, E_2)$ term enters the right-hand side, which matches the sequential covariance form (\Cref{eq:var-sequential-cov}).

\subsection{Closed Forms and Limits for the Makespan Variance}
\label{app:variance-max}

Two cases of $\mathrm{Var}(\max(T_1, T_2))$ admit a quick treatment.
First, when $T_1$ and $T_2$ are independent Gaussian with means $\mu_i$ and variances $\sigma^2_{T_i}$, $\mathrm{Var}(\max(T_1, T_2))$ has a known closed form in terms of the standard normal CDF and PDF.
Second, when the runtimes are well separated, $|\mu_1 - \mu_2| \gg \sigma_{T_1} + \sigma_{T_2}$, the longer element dominates the maximum and $\mathrm{Var}(\max(T_1, T_2)) \approx \sigma^2_{T_{i^\star}}$ where $i^\star = \arg\max_i \mu_i$.
Under $\parop$, the parallel variance in this regime simplifies to $\sigma^2_{E_{i^\star}} + \sigma^2_{d_j}$, where $j$ is the index of the shorter element: the longer element's energy variance is propagated in full, and only the shorter element's dynamic-energy uncertainty contributes.
Under $\parml$, the makespan fluctuation holds both devices, so the regime instead gives $(s_1 + s_2)^2\, \sigma^2_{T_{i^\star}} + \sigma^2_{d_{i^\star}} + \sigma^2_{d_j}$.
This is the common regime in pipeline-parallel training and in compute-overlapping-communication patterns, where one path is reliably longer than the other.

\subsection{Distributional Treatment}
\label{app:variance-distributional}

If full distributions (not just mean and variance) are available, we compute the composition via convolution (for sums under sequential composition) or order statistics (for $\max$ under parallel composition), both exact for independent components; dependent components require their joint distribution.
This yields the exact composed distribution, from which we can extract any quantile.
The mean-and-variance treatment in the main text is the default; the distributional treatment is a refinement for cases where tail behavior matters (e.g., worst-case energy bounds for safety-critical systems).

\section{Refinement--Coarsening Consistency}
\label{app:refinement}

\begin{property}[Refinement--Coarsening Consistency, full statement]
Let element~$e$ have directly measured reduced signature $S_e = (e, \theta, \kappa, \omega, \phi_{\text{null}}, s_e, T_e, d_e, \sigma^2_{T_e}, \sigma^2_{E_e}, M)$ (\Cref{prop:reduction}).
Let $e$ be refined into sub-elements $e_1, \ldots, e_n$ composed via operators $\mathrm{op}_1, \ldots, \mathrm{op}_{n-1}$, yielding a composed reduced signature $\hat{S}_e$ with components $(\hat{s}, \hat{T}, \hat{d},\allowbreak \hat{\sigma}^2_T, \hat{\sigma}^2_E, M')$.
Then, assuming mutual context-insensitivity within uncertainty:
\begin{equation}
    |E_e - \hat{E}_e| \leq \sum_i m_i \, c_\alpha \, \sigma_{E_{e_i}} + c_\alpha \cdot \sqrt{\sigma^2_{E_e} + \hat{\sigma}^2_E}
\end{equation}
at confidence level~$(1 - \alpha)$, with the bias term dropping to zero when each $e_i$'s signature is instead correctly characterized at every entry and runtime context the refinement realizes.
\end{property}

\begin{proof}[Proof sketch]
$E_e$ and $\hat{E}_e$ are two estimates of the same physical quantity.
The direct estimate is unbiased up to measurement error; the composed estimate carries the reduction bias, bounded by $\sum_i m_i c_\alpha \sigma_{E_{e_i}}$ (\Cref{app:reduction-accumulation}) under mutual context-insensitivity, and zero under correct characterization at every realized context.
After subtracting that bias, the difference is a zero-mean random variable with variance bounded by $\sigma^2_{E_e} + \hat{\sigma}^2_E$ when the two estimates share no measurements; shared measurements add a covariance term.
The Gaussian approximation with $c_\alpha$ the two-sided Gaussian multiplier, consistent with \Cref{def:history-insensitive}, then yields the stated bound; it also holds distribution-free via Chebyshev's inequality at the cost of a larger $c_\alpha = 1/\sqrt{\alpha}$.
A violation of this bound indicates an incorrect measurement, an element refinement that omits a material interaction (mutual context-sensitivity not absorbed by characterization), or both.
\end{proof}

Refinement and coarsening are inverses in the following sense.
Refinement decomposes $e$ into sub-elements and produces $\hat{S}_e$; coarsening takes a collection of composed sub-elements and produces a composite signature.
Iterated refinement followed by coarsening at the same granularity yields a signature that agrees with the original within propagated uncertainty, by the same argument.

\section{Frontier Composition Under Context-Independence}
\label{app:frontier}

This appendix proves the frontier composition results of \Cref{sec:operating-frontier}.
Fix workload parameters $\theta_1, \theta_2$ and entry context~$\phi_0$, and let each element~$e_i$ carry its implementations~$\kappa$, operating set $\Omega_{e_i}$, and time--energy map per \Cref{sec:frontier-defined}.
Both subsections assume mutual context-insensitivity at every choice of implementations and operating points (\Cref{def:mutual-context-insensitive}), and the Reduction Theorem (\Cref{prop:reduction}) evaluates signatures at~$\phi_0$ throughout.
All frontier identities are therefore exact in the model; physically, every swept composite point carries its propagated uncertainty interval, and a point is discarded as dominated only when intervals separate, per the comparison rule of \Cref{sec:uncertainty}.
Same-device parallel sweeps are constrained to pairs sharing the device-wide components $\omega_{\text{dev}}$ (\Cref{sec:composition-operators}); per-element components and implementations vary freely.
We assume Pareto-minimal points are attained, e.g., the swept sets finite or compact with continuous time--energy maps.

\subsection{Sequential Frontier Composition}
\label{app:frontier-sequential}

This subsection proves \Cref{prop:frontier-sequential}.

\begin{proof}[Proof of \Cref{prop:frontier-sequential}]
Under $\seqop$, durations add and $E_{e_1 \seqop e_2} = E_1(\kappa_1, \omega_1) + E_2(\kappa_2, \omega_2)$, with each element's values independent of the other's choice by mutual context-insensitivity, so each choice of implementations and operating points lands at the coordinate-wise sum of the two elements' swept points.
The standard domination argument applies: if an element's point lies off $\mathcal{F}_{e_i}$, replacing it with a dominating point improves the composite in at least one coordinate and worsens neither, contradicting Pareto-minimality; hence Pareto-minimal composite points arise from points of $\mathcal{F}_{e_1}$ and $\mathcal{F}_{e_2}$, and the composite's achievable set is $\mathcal{F}_{e_1} \boxplus \mathcal{F}_{e_2}$ up to discarding Pareto-dominated points.
Discarding them yields $\mathcal{F}_{e_1 \seqop e_2} = \mathrm{Pareto\text{-}min}\bigl(\mathcal{F}_{e_1} \boxplus \mathcal{F}_{e_2}\bigr)$.
\end{proof}

Associativity of frontier composition under~$\seqop$ follows from associativity of~$\boxplus$, with Pareto-dominated points discarded early or late to the same effect; the result extends to any finite number of mutually context-insensitive elements by induction.

\subsection{Parallel Frontier Coupling}
\label{app:frontier-parallel}

Under both parallel operators the two frontiers couple through the shared makespan $T_{\max} = \max(T_1, T_2)$.
Same-device composition draws the device's static power once (\axname{P2}), so, at each pair of implementations,
\begin{equation}
\label{eq:parallel-energy-omega}
    E_{e_1 \parop e_2}(\omega_1, \omega_2) = s(\omega_{\text{dev}}) \cdot T_{\max} + d_1(\omega_1) + d_2(\omega_2), \qquad \omega_1, \omega_2~\text{sharing}~\omega_{\text{dev}};
\end{equation}
across devices, $\parml$ requires equal durations, so an idle element pads the earlier-finishing device (\Cref{sec:cross-device-parallel}); the same makespan rule follows with each device drawing its own static power and no shared-$\omega_{\text{dev}}$ constraint, $E = \bigl(s_1(\omega_1) + s_2(\omega_2)\bigr)\, T_{\max} + d_1(\omega_1) + d_2(\omega_2)$, taking each device's idle draw $\bar{s}_i$, its static power in the ready state it idles in, equal to the active draw ($\bar{s}_i = s_i$); the slack proof below keeps the two draws distinct.
In both cases the composite frontier is obtained by sweeping the elements' implementations and operating points and keeping the Pareto-optimal points; the shared makespan couples the two sweeps, so unlike the sequential case the composition is not a Minkowski sum.

\begin{proof}[Proof of \Cref{prop:slack-conversion}]
On a shared device, the composite is $s \cdot T_{\max} + d_1 + d_2$ with $s$ the device's static power (\axname{P2}); the move leaves the device-wide operating point, hence $s$ and $e_2$'s signature, unchanged, and the makespan unchanged since $T_1' \leq T_{\max}$, so $\Delta E = d_1' - d_1 \leq 0$ by condition~(i).
On separate devices, per the accounting of \Cref{sec:composition-operators}, each device draws its element's static power over the element's duration and $\bar{s}_i$ through its idle element over the remainder of the makespan; the idle element's operating point is schedule-chosen, so the move leaves $\bar{s}_i$ unchanged.
This regroups as $E = \tilde{E}_1 + \tilde{E}_2 + (\bar{s}_1 + \bar{s}_2) \cdot T_{\max}$ with $\tilde{E}_i = d_i + (s_i - \bar{s}_i)\, T_i$ element~$i$'s energy above what its device would draw idling over the same period, the decomposition of Perseus~\cite{perseus:sosp24}.
The move changes only $\tilde{E}_1$ at unchanged $T_{\max}$, and with $\delta = T_1' - T_1$,
\begin{equation*}
    \Delta E \;=\; \underbrace{d_1' - d_1}_{\text{(i) dynamic saving}} \;+\; \underbrace{(s_1' - s_1)\, T_1}_{\text{(ii) static saving, active period}} \;+\; \underbrace{(s_1' - \bar{s}_1)\, \delta}_{\text{(iii) window exchange}},
\end{equation*}
each term non-positive under its condition; the window exchange fills time formerly idling at $\bar{s}_1$ with computation drawing $s_1'$.
The sign of (iii) is the race-to-idle decision: on hardware whose idle state draws less than the slowed active state, racing to finish and idling is the cheaper schedule and slack filling costs energy.
\end{proof}

%% file: energy-calculus.bbl
\begin{thebibliography}{10}

\bibitem{nvml}
{{NVIDIA Management Library (NVML)}}.
\newblock \url{https://developer.nvidia.com/nvidia-management-library-nvml}.

\bibitem{zeus:github}
{Zeus}.
\newblock \url{https://github.com/ml-energy/zeus}.

\bibitem{amd-ppr-17h}
{Advanced Micro Devices, Inc.}
\newblock Processor programming reference ({PPR}) for {AMD} family 17h model
  20h, revision {A1} processors.
\newblock \url{https://docs.amd.com/v/u/en-US/55772-A1-PUB_3.08}, 2021.

\bibitem{baccelli1992synchronization}
Fran\c{c}ois Baccelli, Guy Cohen, Geert~Jan Olsder, and Jean-Pierre Quadrat.
\newblock {\em Synchronization and Linearity: An Algebra for Discrete Event
  Systems}.
\newblock Wiley New York, 1992.

\bibitem{cbre2025}
{CBRE}.
\newblock Global data center trends 2025.
\newblock
  \url{https://www.cbre.com/insights/reports/global-data-center-trends-2025},
  2025.

\bibitem{perseus:sosp24}
Jae-Won Chung, Yile Gu, Insu Jang, Luoxi Meng, Nikhil Bansal, and Mosharaf
  Chowdhury.
\newblock Reducing energy bloat in large model training.
\newblock In {\em SOSP}, 2024.

\bibitem{mlenergy-benchmark:neuripsdb25}
Jae-Won Chung, Jeff~J. Ma, Ruofan Wu, Jiachen Liu, Oh~Jun Kweon, Yuxuan Xia,
  Zhiyu Wu, and Mosharaf Chowdhury.
\newblock The {ML.ENERGY} benchmark: Toward automated inference energy
  measurement and optimization.
\newblock In {\em NeurIPS D\&B}, 2025.

\bibitem{dtct-complexity}
Prabuddha De, E.~James Dunne, Jay~B. Ghosh, and Charles~E. Wells.
\newblock Complexity of the discrete time-cost tradeoff problem for project
  networks.
\newblock {\em Operations Research}, 45(2):302--306, 1997.

\bibitem{intel-sdm-vol3b}
{Intel Corporation}.
\newblock Intel 64 and {IA-32} architectures software developer's manual,
  volume {3B}: System programming guide, part 2.
\newblock
  \url{https://www.intel.com/content/www/us/en/developer/articles/technical/intel-sdm.html},
  2026.

\bibitem{accelwattch-micro21}
Vijay Kandiah, Scott Peverelle, Mahmoud Khairy, Junrui Pan, Amogh Manjunath,
  Timothy~G Rogers, Tor~M Aamodt, and Nikos Hardavellas.
\newblock {AccelWattch}: A power modeling framework for modern gpus.
\newblock In {\em MICRO}, 2021.

\bibitem{bloombergnef25}
Helen Kou.
\newblock Power for {AI}: Easier said than built.
\newblock
  \url{https://about.bnef.com/insights/commodities/power-for-ai-easier-said-than-built/},
  2025.

\bibitem{leboudec-thiran}
Jean-Yves {Le Boudec} and Patrick Thiran.
\newblock {\em Network Calculus: A Theory of Deterministic Queuing Systems for
  the Internet}.
\newblock Springer Berlin Heidelberg, 2002.

\bibitem{forecasting-asplos25}
Seonho Lee, Amar Phanishayee, and Divya Mahajan.
\newblock Forecasting {GPU} performance for deep learning training and
  inference.
\newblock In {\em ASPLOS}, 2025.

\bibitem{nvidia-a100-datasheet}
{NVIDIA}.
\newblock {NVIDIA A100 Tensor Core GPU} datasheet, 2021.

\bibitem{selinger:graphical}
Peter Selinger.
\newblock A survey of graphical languages for monoidal categories.
\newblock In Bob Coecke, editor, {\em New Structures for Physics}, volume 813
  of {\em Lecture Notes in Physics}, pages 289--355. Springer, 2010.

\bibitem{inferencemax-blog25}
SemiAnalysis.
\newblock {InferenceMAX}: Open source inference benchmarking.
\newblock
  \url{https://newsletter.semianalysis.com/p/inferencemax-open-source-inference},
  2025.

\bibitem{dtct-apx-hardness}
Ola Svensson.
\newblock Hardness of vertex deletion and project scheduling.
\newblock In {\em Approximation, Randomization, and Combinatorial Optimization.
  Algorithms and Techniques}, pages 301--312, Berlin, Heidelberg, 2012.
  Springer Berlin Heidelberg.

\bibitem{eia-utility-data}
{U.S. Energy Information Administration (EIA)}.
\newblock Capital cost and performance characteristics for utility-scale
  electric power generating technologies.
\newblock
  \url{https://www.eia.gov/analysis/studies/powerplants/capitalcost/pdf/capital_cost_AEO2025.pdf},
  2024.

\bibitem{williams2009roofline}
Samuel Williams, Andrew Waterman, and David Patterson.
\newblock Roofline: an insightful visual performance model for multicore
  architectures.
\newblock {\em Communications of the ACM}, 52(4):65--76, 2009.

\bibitem{kareus:osdi26}
Ruofan Wu, Jae-Won Chung, and Mosharaf Chowdhury.
\newblock Kareus: Joint reduction of dynamic and static energy in large model
  training.
\newblock In {\em OSDI}, 2026.

\bibitem{zeus:nsdi23}
Jie You, Jae-Won Chung, and Mosharaf Chowdhury.
\newblock Zeus: Understanding and optimizing {GPU} energy consumption of {DNN}
  training.
\newblock In {\em USENIX NSDI}, 2023.

\bibitem{nanoflow-osdi25}
Kan Zhu, Yufei Gao, Yilong Zhao, Liangyu Zhao, Gefei Zuo, Yile Gu, Dedong Xie,
  Tian Tang, Qinyu Xu, Zihao Ye, Keisuke Kamahori, Chien-Yu Lin, Ziren Wang,
  Stephanie Wang, Arvind Krishnamurthy, and Baris Kasikci.
\newblock {NanoFlow}: Towards optimal large language model serving throughput.
\newblock In {\em OSDI}, 2025.

\end{thebibliography}
